 \documentclass[11pt]{article}
 \usepackage{fullpage}
 \usepackage{latexsym}

\usepackage{amssymb,amsmath,amsfonts}
\usepackage{amsmath}
\usepackage{graphicx}
\usepackage{times}
 \usepackage{ifpdf}

\usepackage{algorithm}
\usepackage{algorithmicx}
\usepackage[noend]{algpseudocode}

\usepackage{epsfig}
 \usepackage{times}

\newcommand{\ignore}[1]{}
\newcommand{\notinproc}[1]{}
\newcommand{\onlyinproc}[1]{#1}

\newtheorem{thm}{Theorem}[section]
\newtheorem{theorem}{Theorem}[section]

\newtheorem{lemma}[thm]{Lemma}

\newtheorem{corollary}[thm]{ Corollary}

 \newcommand{\qed}{\hfill \rule{1ex}{1ex}\medskip\\}
 \newenvironment{proof}{\paragraph{Proof}}{\qed}

\def\lth{\mbox{$\ell^{\mathrm{th}}$}}
\def\range{\mbox{{\sc rg}}}

\def\var{\mbox{\sc var}}

\def\E{{\textsf E}}

\def\vecv{\boldsymbol{v}}
\def\vecz{\boldsymbol{z}}

\begin{document}

\title{What You Can Do with Coordinated Samples}

   \ignore{
\numberofauthors{2}
\author{
\alignauthor Edith Cohen\\
       \affaddr{Microsoft Research}\\
        \affaddr{Mountain View, CA, USA}\\
       \email{edith@cohenwang.com}
\alignauthor  Haim Kaplan \\
       \affaddr{School of Computer Science}\\
       \affaddr{Tel Aviv University}\\
       \affaddr{Tel Aviv, Israel}\\
       \email{haimk@cs.tau.ac.il}
}
   }

\author{
Edith Cohen\thanks{Microsoft Research, SVC
{\tt edith@cohenwang.com}
} $^\dagger$ \and Haim Kaplan \thanks{
    The Blavatnik School of Computer Science, Tel Aviv University, Tel
    Aviv, Israel. {\tt haimk@cs.tau.ac.il}}
}

\date{}

  \ignore{
\author{
Edith Cohen\inst{1,2}
\and Haim Kaplan \inst{2}
}

\institute{Microsoft Research, SVC, USA
\email{edith@cohenwang.com}
 \and    The Blavatnik School of Computer Science, Tel Aviv University, Israel. \email{haimk@cs.tau.ac.il}}
 }


 \maketitle
\begin{abstract}
 {\small
 Sample coordination, where similar instances have similar samples, was proposed by statisticians four decades ago as
  a way to maximize overlap in repeated surveys.  Coordinated sampling
  had been since used for summarizing massive data sets.  

  The usefulness of a sampling scheme hinges on the scope and accuracy
  within which queries posed over the original data can be answered
  from the sample.  We aim here to gain a fundamental understanding of
  the limits and potential of coordination.  Our main result is a
  precise characterization, in terms of simple properties of the
  estimated function, of queries for which estimators with desirable
  properties exist.  We consider unbiasedness, nonnegativity, finite
  variance, and bounded estimates.

  Since generally a single estimator can not be optimal (minimize
  variance simultaneously) for all data, we propose {\em variance
    competitiveness}, which means that the expectation of the square on any data is not too
  far from the minimum one possible for the data.  Surprisingly
  perhaps, we show how to construct, for any function for which an
  unbiased nonnegative estimator exists, a variance competitive
  estimator. \footnote{This is the full version of a RANDOM 2013 paper}


}
\end{abstract}


\section{Introduction} \label{intro:sec}

Many data sources, including IP or Web traffic logs from different
time periods or locations, measurement data, snapshots of data
depositories that evolve over time, and document/feature and
market-basket data, can be viewed as a collection of {\em
instances}, where each instance is an assignment of numeric values from some
set $V$ to a set of items (the set of items is the same for different
instances but the value of each item changes).

When the data is too massive to manipulate or even store in full,
it is useful to obtain and work with a random sample of each instance.  
Two common sampling schemes are
Poisson sampling (each item is sampled
  independently with probability that depends only on its value)  and
  bottom-$k$ (order) sampling.
The samples 
are efficient to compute, also when instances are presented as
streams or are distributed across multiple servers. 
It is convenient to specify these sampling schemes 
 through a {\em rank function}, $r:[0,1]\times V\rightarrow \mathbb{R}$, which maps seed-value pairs to
a number $r(u,v)$ that is  non-increasing with $u$ and
non-decreasing with $v$.
For each item $h$ we draw a {\em seed}
$u(h) \sim U[0,1]$ uniformly at random and compute the rank value
$r(u(h),v(h))$, where $v(h)$ is the value of $h$.  With Poisson
sampling, item $h$ is sampled $\iff$ $r(u(h),v(h)) \geq T(h)$,
where $T(h)$ are fixed thresholds, whereas
a bottom-$k$ sample includes the $k$ items with highest ranks.\footnote{The
term bottom-$k$ is due to historic usage of the
inverse rank function and lowest $k$
  ranks~\cite{Rosen1972:successive,Rosen1997a,bottomk07:ds,bottomk:VLDB2008,CK:sigmetrics09}}
Poisson PPS
samples (Probability Proportional to Size~\cite{Hajekbook1981}, where
each item is included with probability proportion to its value) are
obtained using the rank function $r(u,v) = v/u$ and a fixed $T(h)$
across items.
 Priority (sequential Poisson)
samples~\cite{Ohlsson_SPS:1998,DLT:jacm07,Szegedy:stoc06} are
bottom-$k$ samples utilizing the PPS ranks $r(u,v) = v/u$ and successive
weighted sampling without replacement ~\cite{Rosen1972:successive,ES:IPL2006,bottomk07:ds} corresponds to bottom-$k$
samples with the rank function $r(u,v)=-v/\ln(u)$.

Samples of different instances are {\em coordinated} when
the set of random seeds $u(h)$ is shared across instances.  Scalable sharing of
seeds when instances are dispersed in time or location is facilitated through
random hash functions $u(h) \gets H(h)$, where the only requirement
for our purposes is uniformity and pairwise independence.
Figure~\ref{example1:fig} contains an example data set of two instances
and the PPS sampling probabilities of each item in each instance, 
and illustrates how to coordinate the samples.
Note that the sample of one instance does not depend on values
 assumed in other instances, which is important for scalable deployment.


\smallskip
\noindent
{\bf Why coordinate samples?}
Sample coordination was proposed
in 1972 by Brewer, Early, and  Joice~\cite{BrEaJo:1972}, as a method to
maximize overlap and therefore minimize overhead in repeated
surveys~\cite{Saavedra:1995,Ohlsson:2000,Rosen1997a}:  The values of
items change, and therefore there is a new set of PPS sampling probabilities.  
With coordination, the sample of the new instance is as similar
as possible to the previous sample, and therefore the number of items that need to be
surveyed again is minimized.
\ignore{The underlying set of weights (amount of traffic on different road
segments) may change over time, but surveys are expensive, and
therefore we want the new sample to be both a true PPS sample
according to the new weights and at the same time, overlap as much as possible
with the previous sample (in order to minimize overhead of
surveying new road segments).}  Coordination was subsequently used
to facilitate efficient processing of large data sets.
\ignore{
Coordination enables efficient sampling over unique items when there are
multiple occurrences, which is useful in stream processing and distributed settings.}
Coordinated samples of instances are used as
synopses which facilitate efficient estimation of
multi-instance functions
  such as distinct counts (cardinality of set unions), sum of maxima, and similarity
\cite{Broder:CPM00,BRODER:sequences97,ECohen6f,CoWaSu,MS:PODC06,GT:spaa2001,Gibbons:vldb2001,BCMS:ton2004,DDGR07,bottomk:VLDB2008,BHRSG:sigmod07,HYKS:VLDB2008,CK:sigmetrics09,multiw:VLDB2009}.
Estimates obtained over coordinated samples are much more accurate than possible with independent samples.
Used this way, coordinated sampling can be casted as a form of
Locality Sensitive Hashing (LSH)
\cite{IndykMotwani:stoc98,GIM:vldb99,indyk:stable}.
Lastly, coordinated samples 
can sometimes be obtained much more efficiently than independent
samples.  One example is computing samples of
the $d$-neighborhoods of all nodes in a graph
\cite{ECohen6f,CoKa:jcss07,MS:PODC06,bottomk07:ds,bottomk:VLDB2008,ECohenADS:2013}.
Similarity queries between neighborhoods are 
useful in the analysis of 
massive graph datasets such as social networks or Web graphs.
\ignore{
Broder \cite{Broder:CPM00,BRODER:sequences97} used coordinated samples
of features in documents to estimate Jaccard similarity and quickly identify similar
documents .  Gibbons~\cite{Gibbons:vldb2001} used coordination  for distinct counts
and Gibbons and Tirthapura~\cite{GT:spaa2001} looked at L$_1$
and union estimates.
 In \cite{relations:CK08,CK:sigmetrics09,
   multiw:VLDB2009} we provided generic constructions
 of estimators for basic multi-instance functions, including quantiles sums and
 $L_1$ difference.  
}

 Our aim here is to study the potential and limitations of estimating multi-instance functions
from coordinated samples of instances.
The same set of samples can be used to estimate multiple
queries, including queries that are specified only after 
sampling is performed.  The sample may also be primarily used
ro estimate subset statistics over one instance as a time, such as sums, averages,
and moments.  While we focus here on queries that span multiple
instances,
we keep this in mind, and  therefore do not aim for a 
sampling scheme optimized for a particular query (although some of our
results can be applied this way), but rather, to 
 optimize the estimator given the sampling scheme and query.



\begin{figure}
{\small
\begin{tabular}{l|| llllllll }
 items: &  1  &  2   &   3  &   4   & 5   &  6  & 7  &  8  \\
\hline
Instance1: &  1  & 0    &  4   & 1   & 0   & 2   & 3   & 1 \\
Instance2: &   3 & 2   &  1   &  0  &  2  & 3   &  1  &  0 \\
\hline
 \multicolumn{9}{c}{PPS sampling probabilities for T=4 (sample of
   expected  size 3):}\\
\hline
Instance1: &  0.25   & 0.00    &  1.00   & 0.25   & 0.00   & 0.50   &  0.75   & 0.25 \\
Instance2: &   0.75 &  0.50   &  0.25   &  0.00  &  0.50  & 0.75   &  0.25  &  0.00 
\end{tabular}}
\caption{Two instances with 8 items and 
  respective PPS sampling probabilities for threshold value $4$, so
item with value $v$ is sampled with
 probability $\min\{1, v/4\}$.
   To obtain two coordinated
  PPS samples of the instances, we associate an independent $u(i)\sim U[0,1]$ with each item $i\in
  [8]$.  We then sample $i\in [8]$ in instance $h\in [2]$ 
if and only if $u(i) \leq v_h(i)/4$, where $v_h(i)$ is the value of $i$ in
instance $h$.
 When coordinating the samples this way, we make them as similar as
  possible. In the example,
 item $1$ will always (for any drawing of seeds)
 be sampled in instance $2$ if
  it is sampled in instance $1$ and vice versa for item $7$. \label{example1:fig}}
\end{figure}

\smallskip
\noindent
{\bf Sum aggregates:}
Most queries in the above examples can be casted as
{\em sum aggregates}
over selected items $h$ of an {\em item function} $f(\vecv)$ applied to the item
weight tuple in different instances $\vecv(h)=(v_1(h),v_2(h),\cdots)$.
In particular, distinct count (set union) is a sum
aggregate of $OR(\vecv)$, max-sum aggregates $\max(\vecv)=\max_i v_i$, min-sum aggregates $\min(\vecv)=\min_i v_i$, and
$L_p^p$ ($p$th power of L$_p$-difference) is the sum aggregate of the 
exponentiated range function
$\range_p(\vecv)=|\max(\vecv)-\min(\vecv)|^p$.
In our example of Figure~\ref{example1:fig},
$L_2^2$ of items $[4]$ is $(1-3)^2 +(2-0)^2 +(4-1)^2 + (1-0)^2 =
18$, and is computed by summing the item function
$\range_2(v_1,v_2)=(v_1-v_2)^2$ over these items.  The $L_1$ of items
$\{1,3\}$ is $|1-3|+ | 4-1|=5$, using the item function
$\range(v_1,v_2)=|v_1-v_2|$, and the
max-sum 
of items $\{6,7,8\}$ is $\max\{2,3\}+\max\{3,1\}+ \max\{1,0\}=7$,
which uses the item function $\max\{v_1,v_2\}$.
Other queries in our examples which are not in sum aggregate form can be approximated
well by sum aggregates: The Jaccard similarity is a ratio
of min-sum and max-sum, and thus the ratio of min-sum estimate
with small additive error and a  max-sum estimate with
small relative error~\cite{multiw:VLDB2009} is a good
Jaccard similarity estimator.
  The $L_p$ difference is the $p$th root of $L_p^p$, and can be estimated well
by taking the $p$th root of a good estimator for $L_p^p$, which is a
sum aggregate of $\range_p$.
 When the domain of query results is nonnegative, as is the case in
 the examples,  the common practice, which we follow,  is to restrict the estimates, which often are plugged in instead of
 the exact result, to be nonnegative as well.


\smallskip
\noindent
{\bf Sum estimators -- one item at a time:}
To estimate a sum aggregate, we can use a linear
estimator which is the sum of single-item estimators,
estimating the item function $f(\vecv(h))$ for each selected item $h$.
We refer to such an estimator as a {\em sum estimator}.
When the single-item estimators are unbiased, from linearity of
expectation, so is the sum estimate.  When 
the single-item estimators are unbiased and sampling of different
items is pairwise
independent (respectively, negatively correlated, as with bottom-$k$ sampling), the variance of the
sum is (resp., at
most) the sum of variances of the
single-item estimators.  Therefore,
the relative error of the sum estimator decreases with the number of
selected items we aggregate over.
We emphasize  that unbiasedness of the single-item estimators
(together with pairwise independence or negative correlations between items)
is critical for good estimates of the sum aggregate
since a variance component that is due to bias ``adds up'' with
aggregation whereas otherwise the relative error ``cancels out'' with aggregation.
\footnote{More concretely, 
the sample of any particular item is likely to have all or
most entries missing, in which case,
the variance of any nonnegative unbiased single-item estimator is 
$\Omega((1/p)f(\vecv)^2)$, where $p$ is the probability that the outcome reveals
``enough'' on $f(\vecv)$.
In such a case, a fixed $0$ estimate (which is biased) will have
variance $f(\vecv)^2$, which is lower than 
the variance of any unbiased estimator when $p$ is small,
but also clearly useless and will result in large error for the
sum aggregate.
To understand this more precisely, let $p$ be the probability that the
outcome provides a lower bound of at least $f(\vecv)/2$. To
be nonnegative, the contribution to the expectation from the other
$(1-p)$ portion of outcomes can not excede
$f(\vecv)/2$, so the remaining contribution of the $p$ portion of
outcomes must be at least $f(\vecv)/2$, giving the lower bound on the
variance.
}
\ignore{
The point of requiring unbiasedness is that it allows for meaningful single-item
estimates  where  the error cancels out.
(We note that we do not insist on unbiasedness of aggregate estimators:   the Jaccard similarity and $L_p$ estimators (when $p\not=
1$) obtained through sum aggregates as explained above are biased even when the estimates for the
underlying sum aggregates are not, but these are the ``end products'' of
our estimation so
we only worry about bias  in terms of its contribution
to the relative error of the estimate.)

}
 The Horvitz-Thompson (HT) estimator~\cite{HT52} is a
classic sum estimator which is 
unbiased and nonnegative.
To estimate $f(v)$, the HT estimator outputs $0$ when the
value is not sampled and the inverse-probability estimate $f(v)/p$
when the value is sampled, where $p$ is the sampling probability.
The HT estimator
is applicable to some multi-instance functions \cite{multiw:VLDB2009}.
\ignore{When there is only one instance, the estimator is applied to one entry at 
a time, and outputs an estimate of $f(v)/p$ when the
entry is sampled and $0$ otherwise (when entry is sampled we know $v$
and therefore can deduce the probability $p=\sup_{u\in (0,1]} r(u,v)
\geq T(h)$ that the entry is sampled and thus compute the HT
estimate.).  The HT estimator has minimum variance for each entry
and therefore (on a single instance) is the sum estimator with minimum
variance for all data.}




 From here on, we restrict our attention to estimating item functions $f(\vecv)\geq 0$
where each entry of $\vecv$ is Poisson sampled and focus on unbiased
and nonnegative estimators for $f(\vecv)$. 
We provide the model in detail in Section \ref{model:sec}.
See Appendix
\ref{sumcase:sec} for further discussion on using sum estimators with
bottom-$k$ samples.

\smallskip
\noindent
{\bf The challenge we address:}
Throughout the 40 year period in which coordination was used,
estimators were developed in an ad-hoc manner,
lacking a fundamental understanding of the potential and
limits of the approach.  
Prior work was mostly based on adaptations of the HT estimator for
multiple instances.
The HT estimator  is applicable 
provided that for any $\vecv$
where $f(\vecv)>0$, there is a positive probability for an outcome
that both reveals $f(\vecv)$ and allows us to determine a
probability $p$ for such an outcome.  These conditions are satisfied 
by some basic  functions including $\max(\vecv)$ and
$\min(\vecv)$.  
 There are functions, however, for which the HT estimator is not applicable, but nonetheless, for which nonnegative
and unbiased estimators exist.  Moreover, the HT estimator may not be optimal
even when it is applicable.

As a particular example,  the only $L_p$  difference for which a ``satisfactory'' estimator was
known was the $L_1$ difference \cite{multiw:VLDB2009}.
Stated in terms of estimators for item functions, prior to our work, there was no unbiased and
nonnegative estimator known for $\range_p(\vecv)=|\max(\vecv)-\min(\vecv)|^p$ for any $p\not=1$.  For 
$p=1$, the known estimator used the relation $\range(\vecv)=\max(\vecv)-\min(\vecv)$, separately estimating the maximum and the
minimum and 
showing that when samples are coordinated then the estimate for the
maximum is always at least as large as the one for the minimum and
therefore the difference of the estimates is  never negative.  
But even for this $\range(\vecv)$ estimator,  there was no understanding whether it is
``optimal'' and more so, what optimality even means in this context.
Moreover, the ad hoc construction of this estimator does not extend
even to slight variations, like 
$\max\{v_1-v_2,0\}$, which sum-aggregates to the natural 
``one sided'' $L_1$ difference.

\subsection*{ Contributions highlights}
 
\smallskip
\noindent
{\bf Characterization:}
 We provide 
a complete characterization, in terms of simple properties of the
 function $f$ and the sampling scheme parameters, of when estimators with
the following combinations of properties exist for $f$:
\begin{trivlist}
\item $\bullet$
unbiasedness and nonnegativity.
\item $\bullet$
unbiasedness, nonnegativity, and {\em finite variances}, which means that
for all $\vecv$, the variance given data $\vecv$ is finite.
\item $\bullet$
unbiasedness,
nonnegativity, and
{\em bounded estimates}, which means that for each $\vecv$,
 there is an upper bound on all estimates that can be 
obtained when the data is $\vecv$.  
Bounded estimates implies finite variances, but not vice versa.
\end{trivlist}

\smallskip
\noindent
{\bf The J estimator}
Our characterization utilizes a construction of an estimator, which we
call {\em the J estimator}, which we show has the following properties:
 The J estimator
 is unbiased and nonnegative if and only if an unbiased
 nonnegative estimator exists for $f$.  The J estimator has a finite
 variance for data $\vecv$ or is bounded for data $\vecv$ if and
 only if an (nonnegative unbiased) estimator with the respective property for $\vecv$  exists.

\smallskip
\noindent
{\bf Variance competitiveness:}
Ideally, we would like to have a single estimator, which
  subject to desired properties of the estimator, minimizes the variance for all data.
Such estimators are known in the statistics literature as
UMVUE (uniform minimum
  variance unbiased) estimators \cite{surveysampling2:book}.
Generally however, an (unbiased, nonnegative, linear) estimator with
  minimum variance on all data vectors \cite{Lanke:Metrica1973} may
  not exist.    
Simple examples show that this is also the case for our coordinated
  sampling model even when restricting our attention to particular natural functions.

  We are therefore introducing and aiming for a notion of {\em variance competitiveness}, which means
that for {\em any} data vector, the variance of our estimator is not ``too
far'' from the minimum variance
possible for that vector by a nonnegative unbiased  estimator.
More precisely, 
an estimator is
$c$-competitive if for all data $\vecv$, the expectation of its square is
within a factor of $c$ from the minimum possible for $\vecv$ by an
estimator that is unbiased and nonnegative on all data.
Variance competitiveness bridges the
  gap between the ideal UMVUE estimators (which generally do not exist)
and the practice of estimator selections with no
  ``worst-case''  guarantees.

\smallskip
\noindent
{\bf $\vecv$-optimality:}
To study competitiveness, we need to compare the variance on each data
vector to the minimum possible, and to do so, we need to be able to 
express the ``best possible'' estimates.
  We say that an estimator is {\em $\vecv$-optimal} if
 amongst all estimators that are unbiased and nonnegative on all data,
 it minimizes variance for the data $\vecv$.  
  We express the $\vecv$-optimal estimates, which are
the values a $\vecv$-optimal estimator assumes on outcomes that are
consistent with  data $\vecv$, and the respective
$\vecv$-optimal variance.  We show that the $\vecv$-optimal estimates
 are uniquely defined (almost everywhere).
 The $\vecv$-optimal estimates, however, are inconsistent for
 different data since as we mentioned earlier, it is not generally possible to
 obtain a single unbiased nonnegative estimator that minimizes
 variance for all data vectors.  They do allow us, however, to analyse
the competitiveness of estimators, and in particular, that of the J
estimator.  

\smallskip
\noindent
{\bf Competitiveness of the J estimator:}
We show that the J estimator is competitive.  In particular, this
shows the powerful and perhaps 
surprising
 result that whenever for any particular data vector $\vecv$ there
 exists an estimator with finite variance that is 
 nonnegative
and unbiased on all data, then there is a {\em single} estimator,
that for all data, has expectation of the square that is $O(1)$ of the minimum possible.

\section{Coordinated Shared-Seed Sampling for a Single Item} \label{model:sec}
\medskip
\noindent

Using the terminology in
the introduction,  we are now looking at a single item, its values 
$\vecv=(v_1,\ldots,v_r)$ in different instances, the sampling scheme projected on this
one item, and estimating the item function $f$.

We denote that data domain by ${\bf V}\subseteq \mathbb{R}^r_{\geq 0}$.  
The {\em sampling scheme} is specified by non-decreasing
continuous  functions
$\mbox{\boldmath{$\tau$}}=(\tau_1,\ldots,\tau_r)$ on $[0,1]$ such that
the infimum of the range of $\tau_i$ is at most $\inf_{\vecv\in {\bf
    V}} v_i$.

 The sampling of $\vecv$ is performed by drawing a random value
$u\sim U[0,1]$, which we refer to as the {\em seed}.  We then include
 the $i$th entry in the outcome $S$ if and only if $v_i$ is at least $\tau_i(u)$:
$$i\in S\,  \iff \, v_i\geq \tau_i(u) \ .$$
We use the notation $S(u,\vecv)$  for the outcome of the sampling
scheme applied to data $\vecv$ and seed $u$.

  Sampling is PPS if $\tau_i(u)$ are linear functions:  there is a fixed vector
$\mbox{\boldmath{$\tau^*$}}$ such that $\tau_i(u) \equiv u \tau^*_i$, in which case
entry $i$ is included with probability
$\min\{1,v_i/\tau^*_i\}$.    Our use of the term PPS in this
single-item context is compatible with PPS 
sampling of each instance $i$ using threshold $\tau^*_i$.  
In this case, we look at the ``projected'' sampling scheme
on the $i$th entry of our item's tuple.

Observe that our model assumes {\em weighted} sampling, where the probability that an
entry is sampled depends (and is
non-decreasing) with its
value.   Transiting briefly back to sampling of instances, weighted
sampling results in more accurate estimation of quantities (such as
averages of sums) where larger values contribute more.
It is also important for
boolean domains (${\bf V}=\{0,1\}^r$) when most items have $0$ values
 and in this case, enables us to sample only ``active'' items.

  We assume that the seed $u$ and
the functions $\mbox{\boldmath{$\tau$}}$ are available to the
estimator, and in particular,  treat the seed as provided with the
outcome.
  When an entry is sampled, we know its value and also can compute
the probability that it is sampled.  When an entry is not sampled,
we know that its value is at most $\tau_i(u)$ and we can compute this
upper bound from the seed $u$ and the function $\tau_i$.
Putting this information together, for each outcome $S(u,\vecv)$, we can define
 the set $V^*(S)$  of all data vectors consistent with the outcome.
 This set captures all the information we can glean from the sample on
 the data.
{\small
$$V^*(S)\equiv V^*(u,\vecv) =\{\vecz \mid \forall i\in [r], \left(i\in S\wedge   z_i=v_i\right)
 \, \vee \, \left(i\not\in S \wedge z_i< \tau_i(u)\right)   \}\ .$$
}

\smallskip
\noindent
{\bf Structure of the set of outcomes.}
From the outcome, which is the set of sampled entries and the seed $\rho$, 
we can determine $V^*(u,\vecv)$ also
for all $u\geq \rho$.  We also have that for all $u\geq \rho$ and  $\vecz\in
V^*(\rho,v)$, $V^*(u,\vecz)=V^*(u,\vecv)$.
Fixing $\vecv$,  the sets $V^*(u,\vecv)$ are
non-decreasing with $u$
and the set $S$ of
sampled entries is non-increasing, 
meaning that $V^*(u,\vecv) \subset
V^*(\rho,\vecv)$
and  $S(u,\vecv) \supset S(\rho,\vecv)$ 
 when $u<\rho$.

The containment order of the sets $V^*(S)$ is a tree-like partial
order on outcomes.
 For two outcomes, the sets $V^*(S)$ are either disjoint, and
unrelated in the containment order, or one is fully
contained in another, and succeeds it in the containment order.
The outcome $S(u,\vecv)$ precedes
$S(\rho,\vecv)$ in the containment order if and only if $u>\rho$.  When $V$ is finite,
the containment order is a tree order, as shown in Figure~\ref{vstar_tree:fig}.

\begin{figure} 
\centerline{
\ifpdf
\includegraphics[width=0.3\textwidth]{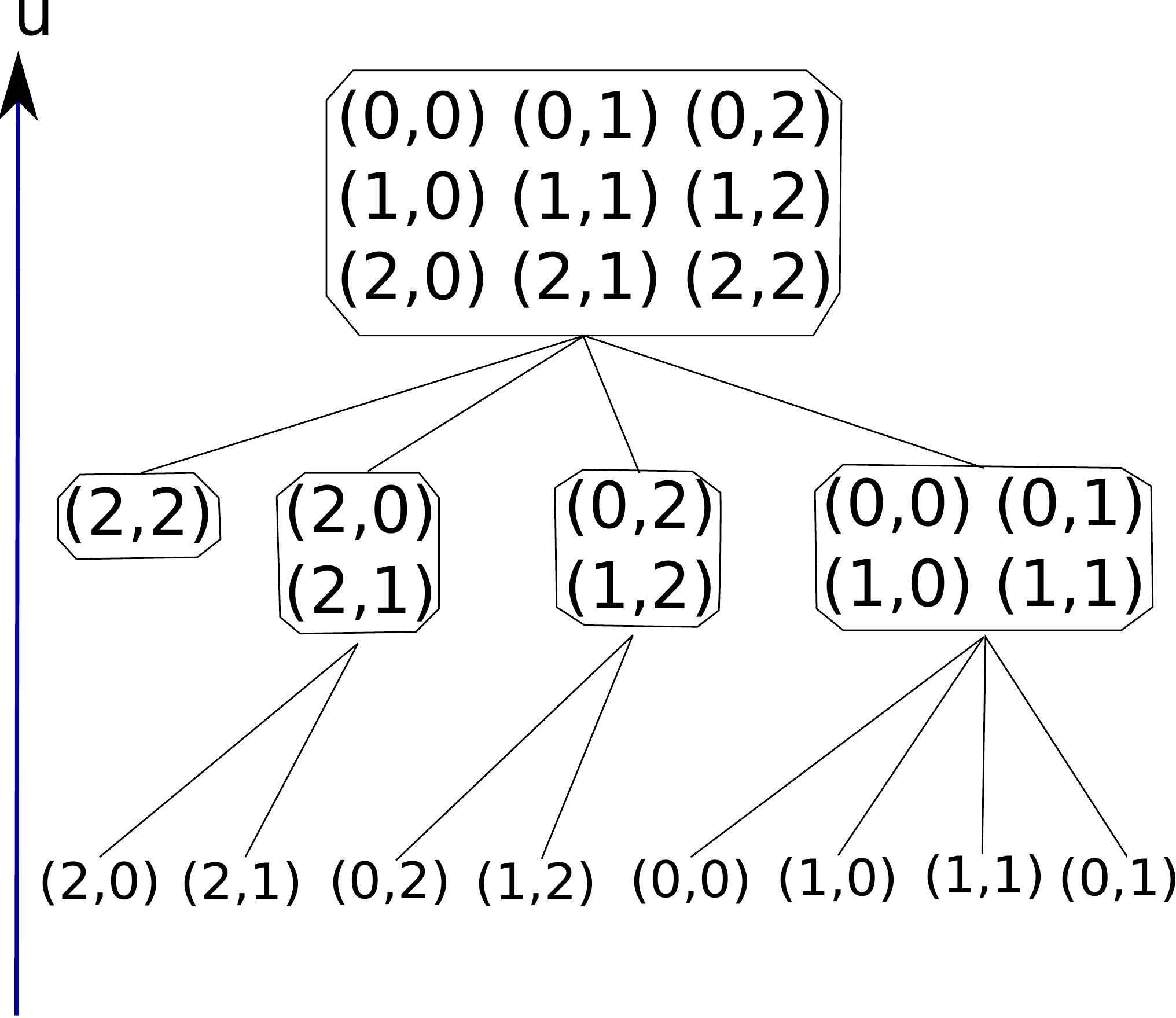}
\else
\fi
}
\caption{Illustration of the containment order on all possible
  outcomes $V^*(S)$. Example has data domain ${\bf V}=\{0,1,2\}\times
  \{0,1,2\}$ and seed mappings $\tau_1=\tau_2 \equiv \tau$.  The root
  of the tree corresponds to outcomes with
$u\in(\tau^{-1}(2),1]$. In this case, the outcome reveals no
information on the data and $V^*(S)$ contains all vectors in ${\bf V}$.
 When $u\in (\tau^{-1}(1),\tau^{-1}(2)]$ the
  outcome identifies entries in the data that are equal to ``2''.  When $u\in
  (0,\tau^{-1}(1)]$, the outcome reveals the data vector.\label{vstar_tree:fig}}
\end{figure}

For $\vecz$ and $\vecv$, the set of all $u$ such that $\vecz \in
S(u,\vecv)$, if not empty,  is a suffix of $(0,1]$ that is open to
the left.
\begin{lemma} \label{openset:lemma}
\begin{eqnarray*}
\lefteqn{\forall \rho\in (0,1]\ \forall \vecv}\\
&& \vecz \in V^*(\rho,\vecv) \implies
\exists \epsilon > 0,\ \forall x\in (\rho-\epsilon,1],\ \vecz\in
V^*(x,\vecv)
\end{eqnarray*}
\end{lemma}
\begin{proof}
Correctness for all $x\in [\rho,1]$ follows from the structure of the
set of outcomes:  Since $V^*(x,\vecv) \supset V^*(\rho,\vecv)$ for
all $x\geq \rho$ then $\vecz\in V^*(\rho,\vecv) \implies \vecz\in V^*(x,\vecv) $.

 Consider now the set $S$ of entries that satisfy $v_i\geq
 \tau_i(\rho)$.
Since $\vecz\in V^*(\rho,\vecv)$, we have $\forall i\in S,
z_i=v_i$
and $\forall i\not\in S,\ \max\{z_i,v_i\}< \tau_i(\rho)$. Since
$\tau_i$ is continuous and monotone,  for all $i\not\in S$,  there
must be an $\epsilon_i >0$
such that $\tau_i(\rho-\epsilon_i)> \max\{z_i,v_i\}$.
We now take $\epsilon = \min_{i\not\in S} \epsilon_i$ to conclude the proof.\onlyinproc{\qed}
\end{proof}

\section{ Estimators and properties}
Let  $f:{\bf V}$ be a function mapping ${\bf V}$ to the nonnegative reals.
An {\em estimator} $\hat{f}$ of $f$ is a numeric function applied
to the outcome.   We use the notation $\hat{f}(u,\vecv)\equiv
\hat{f}(S(u,\vecv))$.
On continuous domains,  an estimator must be (Lebesgue) integrable.
  An estimator is {\em fully  specified} for $\vecv$ if specified on 
a set of outcomes that have
probability $1$ given data $\vecv$.
Two estimators $\hat{f}_1$ and $\hat{f}_2$
are {\em equivalent} if for all data $\vecv$, $\hat{f}_1(u,\vecv)=\hat{f}_2(u,\vecv)$ with probability $1$.
\ignore{
From  the Lebesgue differentiation theorem,
\begin{lemma} \label{lebesguediff}
Two estimators $\hat{f}_1$ and $\hat{f}_2$ are equivalent if and only
if
\begin{equation*}
\forall \vecv \forall \rho\in (0,1],
\lim_{\eta\rightarrow \rho^{-}} \frac{\int_\eta^\rho \hat{f}_1(u,\vecv)du}{\rho-\eta}=
\lim_{\eta\rightarrow \rho^{-}} \frac{\int_\eta^\rho \hat{f}_2(u,\vecv)du}{\rho-\eta}
\end{equation*}
\end{lemma}
}

An estimator $\hat{f}$ is
 {\em nonnegative}  if $\forall S,\ \hat{f}(S)\geq 0$ and is
{\em unbiased} if $\forall  \vecv,\ \E[{\hat f} | \vecv]=f(\vecv)$.
An estimator has
{\em finite variance} on $\vecv$ if $\int_0^1 \hat{f}(u,\vecv)^2
du < \infty$ (the expectation of the square is finite) and is  {\em bounded} on $\vecv$ if $\sup_{u\in (0,1]}
  \hat{f}(u,\vecv) < \infty$.  If a nonnegative estimator is bounded
  on $\vecv$, it also
  has finite variance for $\vecv$.
We say that an estimator is bounded or has finite variances 
if the respective property holds for all $\vecv\in {\bf V}$.

\smallskip
\noindent
{\bf $\vecv$-optimality.}
We say that an unbiased and nonnegative estimator is {\em $\vecv$-optimal},
that is, optimal with respect to a data vector $\vecv$,
if it has minimum variance for $\vecv$.
We refer to the estimates that a $\vecv$-optimal estimator assumes on
outcomes consistent on data $\vecv$ as the {\em $\vecv$-optimal estimates}
and to the minimum variance attainable for $\vecv$ as the {\em $\vecv$-optimal
variance}.

\smallskip
\noindent
{\bf Variance competitiveness.}
An estimator $\hat{f}$ is {\em $c$-competitive} if
$$\forall \vecv,\, \int_0^1 \bigg(\hat{f}(u,\vecv)\bigg)^2du \leq c \inf_{\hat{f}'}\int_0^1 \bigg(\hat{f}'(u,\vecv)\bigg)^2du ,$$
where the infimum is over all unbiased nonnegative estimators $\hat{f}'$
of $f$.
For any unbiased estimator, the expectation of the
square is closely related to the variance:
\begin{align}
\var[\hat{f} | \vecv]&=\int_0^1 (\hat{f}(u,\vecv)-f(\vecv))^2 du 
=\int_0^1 \hat{f}(u,\vecv)^2 du -f(\vecv)^2 \label{var2moment}
\end{align}
When minimizing the expectation of the square, we also minimize the variance.
Moreover,  $c$-competitiveness means that
\begin{equation}  \label{ccompvar}
\forall \vecv,\  \var[\hat{f} | \vecv] \leq c
\inf_{\hat{f}'} \var[\hat{f}' | \vecv] +(c-1)f(\vecv)^2
\end{equation}
for all data vectors $\vecv$ for which  a nonnegative
unbiased estimator with finite variance on $\vecv$ exists,
the variance of the estimator is at most $c$ times the $\vecv$-optimal
variance plus an additive term of $(c-1)$ times $f(\vecv)^2$.

An important remark is due here.  
In the typical scenario, discussed in the introduction, the
sample is likely to provide little or no information on $f(\vecv)$,
the  variance is   $\Omega(f(\vecv)^2)$, and hence competitiveness
as we defined it in terms of the expectation of the square
translates to competitiveness of the variance.  
Otherwise, when for some data in the domain the sample is
likely to reveal the value,  it is not possible to obtain a
universal competitiveness result in terms of variance.
(One such example is $\range$ on PPS samples, looking at data tuples
where the maximum has sampling probability $1$.)  Interestingly, 
for $\range_2$ it is possible to 
get a bounded ratio in terms of variance.  More details
are in the companion experimental paper~\cite{CK_sdiff_arxiv:2012}.



 \section{The lower bound function and its lower hull} \label{lblh:sec}

  For a function $f$, we define the respective {\em lower bound function} 
  $\underline{f}$ and the lower hull function $H_f$.  We then
  characterize, in terms of properties of $\underline{f}$ and $H_f$ when
nonnegative unbiased estimators exists for
$f$ and when such estimators exist that also
have finite variances or are bounded.



\noindent
\medskip
{\bf The lower bound function $\underline{f}(S)$:}
For $Z\subset {\bf V}$,  we
define $\underline{f}(Z) = \inf\{ f(v) \mid v\in Z\}$ to be
the tightest lower bound on the values  of $f$ on $Z$.
We use the notation $\underline{f}(S)\equiv \underline{f}(V^*(S))$,
$\underline{f}(\rho,\vecv) \equiv \underline{f}(V^*(\rho,\vecv))$.
When $\vecv$ is fixed, we use
$\underline{f}^{(\vecv)}(u) \equiv \underline{f}(u,\vecv)$.

From Lemma~\ref{openset:lemma}, we obtain that
$\forall \vecv,\ \underline{f}^{(\vecv)}(u)$ is {\em left-continuous}, that is:
\begin{corollary}  \label{leftcon:lemma}
$\forall \vecv \forall \rho\in(0,1],\ \lim_{\eta\rightarrow \rho^{-}} \underline{f}^{(\vecv)} (\eta)=\underline{f}^{(\vecv)} (\rho)$.
\end{corollary}

\begin{lemma}\label{nonneg:lemma}
A nonnegative unbiased estimator $\hat{f}$ must satisfy
\begin{equation} \label{nonneg:eq}
\forall \vecv, \forall \rho, \int_\rho^1 \hat{f}(u,\vecv)du \leq
\underline{f}^{(\vecv)}(\rho)
\end{equation}
\end{lemma}
\begin{proof}
Unbiased and nonnegative $\hat{f}$ must satisfy
\begin{equation} \label{contbound}
\forall \vecv, \forall \rho\in (0,1].\ \int_{\rho}^1 \hat{f}(u,\vecv)du\leq \int_0^1 \hat{f}(u,\vecv)du=f(\vecv)\ .
\end{equation}
From definition of $\underline{f}$,
for all $\epsilon>0$ and $\rho$, there is
a vector ${\vecz}^{(\epsilon)}\in S(\rho,\vecv)$ such that $f(\vecz^{(\epsilon)})\leq
\underline{f}(\rho,\vecv)+\epsilon$.
Recall that for all $u\geq \rho$, $S(u,\vecv)=S(u,\vecz^{(\epsilon)})$, hence, using, \eqref{contbound},
$$\int_{\rho}^{1} \hat{f}(u,\vecv)du= \int_{\rho}^{1} \hat{f}(u,\vecz^{(\epsilon)})du
\leq f(\vecz^{(\epsilon)})\leq  \underline{f}(\rho,\vecv) + \epsilon\ .$$
Taking the limit as $\epsilon\rightarrow 0$ we obtain
$\int_{\rho}^{1} \hat{f}(u,\vecv)du \leq \underline{f}(\rho,\vecv)\ .$\onlyinproc{\qed}
\end{proof}

\noindent
\medskip
{\bf The lower hull of the lower bound function and $\vecv$-optimality:}
 We denote the function corresponding to the lower boundary of the convex hull (lower hull)
 of $\underline{f}^{(\vecv)}$ by
$H^{(\vecv)}_f$.  
 Our interest in the lower hull is due to the following relation
(The proof is postponed to Section~\ref{vvaropt:sec}):
\begin{theorem}  \label{voptlh}
An estimator $\hat{f}$ is $\vecv$-optimal
 if and only if for $u\in [0,1]$ almost everywhere
$$\hat{f}(u,\vecv) = -\frac{d H^{(\vecv)}_f(u)}{du}\ .$$
Moreover,
when an unbiased and nonnegative estimator exists for $f$, there
also exists, for any data $\vecv$, a
nonnegative and unbiased $\vecv$-optimal estimator.
\end{theorem}
We use the notation $\hat{f}^{(\vecv)}(u) = -\frac{d H^{(\vecv)}_f(u)}{du}$ for the {\em $\vecv$-optimal estimates on outcomes
consistent with $\vecv$}.   
Since the lower bound function is monotone non-increasing,
so is $H^{(\vecv)}_f$, and therefore $H^{(\vecv)}_f$ is
differentiable almost everywhere and $\hat{f}^{(\vecv)}$ is
defined almost everywhere.
Figure~\ref{LBopt:fig} illustrates an example lower bound function and
the corresponding  lower hull.  
\begin{figure}[htbp]
\centerline{
\ifpdf
\includegraphics[width=0.40\textwidth]{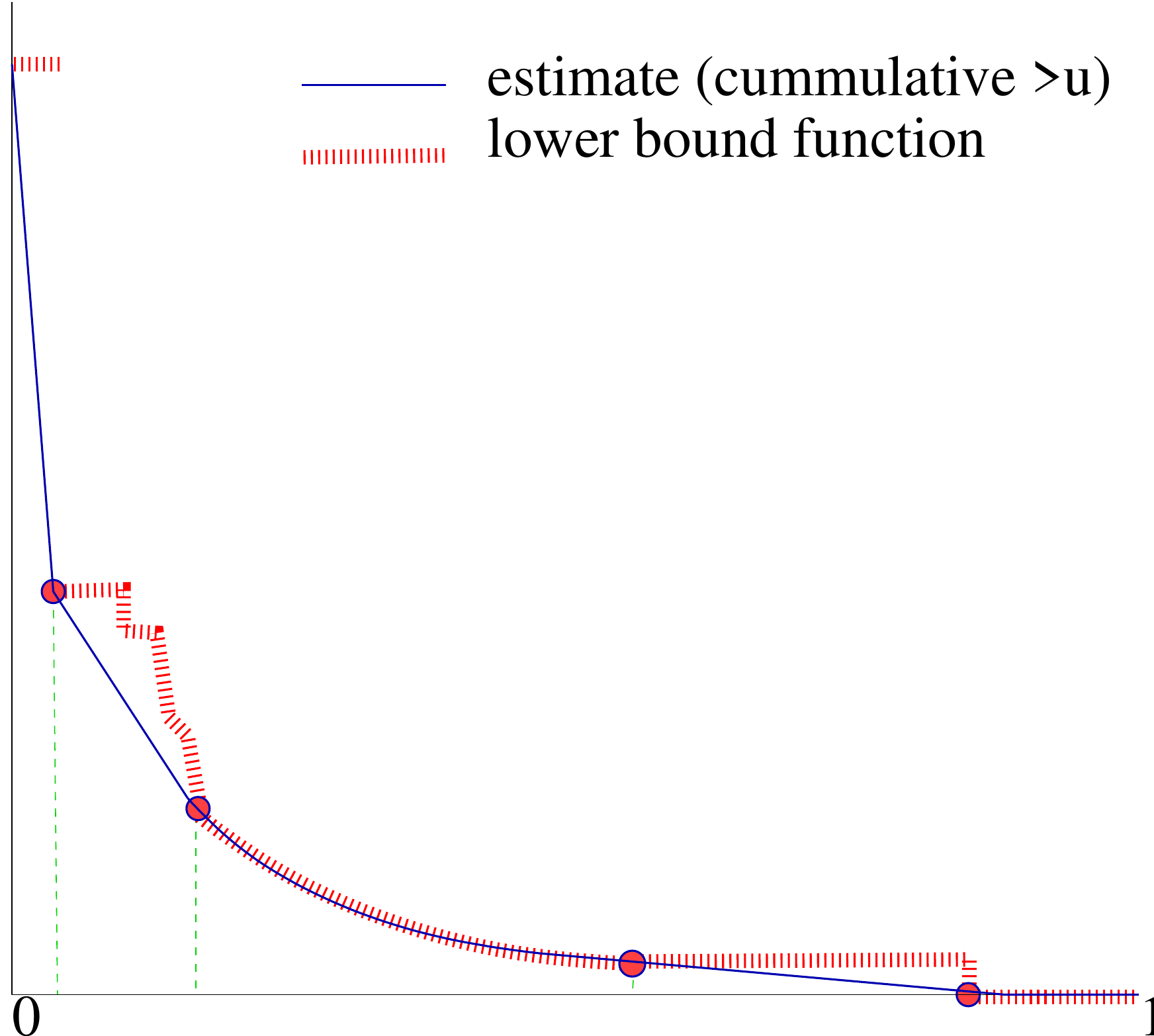}
\else
\epsfig{figure=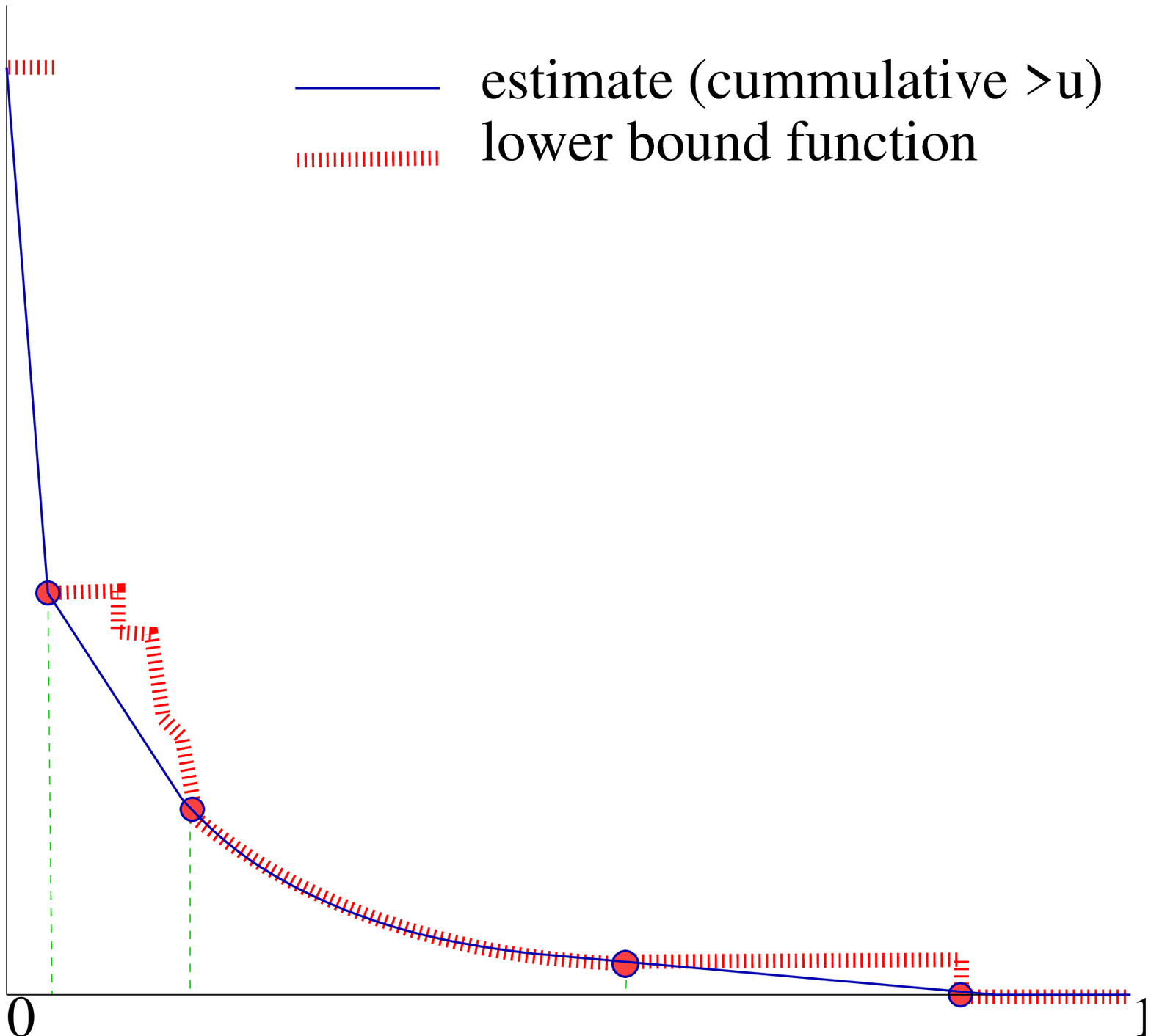,width=0.40\textwidth}
\fi
}
\caption{Lower bound function $\underline{f}^{(\vecv)}(u)$ for $u\in (0,1]$ and
corresponding lower hull $H^{(\vecv)}_f(u)$, which is also
the integral of the nonnegative estimator with minimum
variance on $\vecv$:
$\int_u^1 \hat{f}^{(\vecv)}(x)dx$.
The figure visualizes the lower bound function which is always
left-continuous and monotone non increasing.  The lower hull is
continuous and also monotone non-increasing.
\label{LBopt:fig}}
\end{figure}

 \section{Characterization} \label{charac:sec}

\begin{theorem}  \label{charact:thm}
$f\geq 0$ has an estimator that is
\begin{trivlist}
\item
$\bullet$ unbiased and nonnegative $\iff$
\begin{equation} \label{nec_req}
\forall \vecv\in {\bf V},\,
\lim_{u \rightarrow 0^+} \underline{f}^{(\vecv)}(u) = f(\vecv)\ .
\end{equation}
\ignore{
$\bullet$ unbiased, nonnegative, and finite variance
$\Leftarrow$
 \begin{equation} \label{bounded_var_shared_suff_req}
 \forall \vecv\in {\bf V},\,
 \int_0^1 \bigg(\frac{f(\vecv)-\underline{f}(u,\vecv)}{u}\bigg)^2 du < \infty \ .
 \end{equation}\\
 $\Rightarrow$
 \begin{equation} \label{bounded_var_shared_nec_req}
 \forall \vecv\in {\bf V},\,
 \lim_{u \rightarrow 0^+} \frac{f(\vecv)-\underline{f}^{(\vecv)}(u)^2}{u} < \infty \ .
 \end{equation}
} 
$\bullet$ unbiased, nonnegative, and finite variances $\iff$
\begin{equation} \label{bounded_var_shared_nec_req}
\forall \vecv\in {\bf V},\,
\int_0^1 \bigg(\frac{d H^{(\vecv)}_f(u)}{du} \bigg)^2 du  < \infty \ .
\end{equation}
$\bullet$ unbiased, nonnegative, and bounded $\iff$
\begin{equation} \label{bounded_shared_nec_req}
\forall \vecv\in {\bf V},\,
\lim_{u \rightarrow 0^+} \frac{f(\vecv)-\underline{f}^{(\vecv)}(u)}{u} < \infty \ .
\end{equation}
\end{trivlist}
\end{theorem}

We establish sufficiency in Theorem \ref{charact:thm}  by constructing an
estimator $\hat{f}^{(J)}$ (the J estimator) that is unbiased and nonnegative when
\eqref{nec_req} holds,  bounded when
\eqref{bounded_shared_nec_req} holds, and has finite variances
if \eqref{bounded_var_shared_nec_req} holds.
The proof of the theorem is provided in Section~\ref{charactthmproof}, following
the presentation of the J estimator in the next section.

\section{The J estimator}
Fixing $\vecv$, we define an estimator
$\hat{f}^{(J)}(u,\vecv)$ incrementally, starting with $u=1$ and such
that the value at $u$ depends on values at $u'> u$. We first define
$\hat{f}^{(J)}(u,\vecv)$ for all $u\in (\frac{1}{2}, 1]$ by
$\hat{f}^{(J)}(u,\vecv) = 2\underline{f}(1,\vecv)$.
At each step we consider intervals of the form $(2^{-j-1},2^{-j}]$,
setting the estimate to the same value for all outcomes  $S(u,\vecv)$ for $u\in (2^{-j-1},2^{-j}]$.
Assuming the estimator is defined for $u\geq 2^{-j}$, we extend the
definition to the interval $u\in  (2^{-j-1},2^{-j}]$ as follows. 
\begin{align*}
 \hat{f}^{(J)}(u,\vecv) &= 0 \, \text{, {\bf if}}  \,\,
\underline{f}(2^{-j},\vecv)=\int_{2^{-j}}^1 \hat{f}^{(J)}(u,\vecv)du \\
 \hat{f}^{(J)}(u,\vecv)&=2^{j+1} \bigg(\underline{f}(2^{-j},\vecv)- \int_{2^{-j}}^1 \hat{f}^{(J)}(u,\vecv)du  \bigg)\,  
 \text{, {\bf otherwise}} 
\end{align*}

\begin{figure}[htbp]
\centerline{
\ifpdf
\includegraphics[width=0.47\textwidth]{Jest}
\else
\epsfig{figure=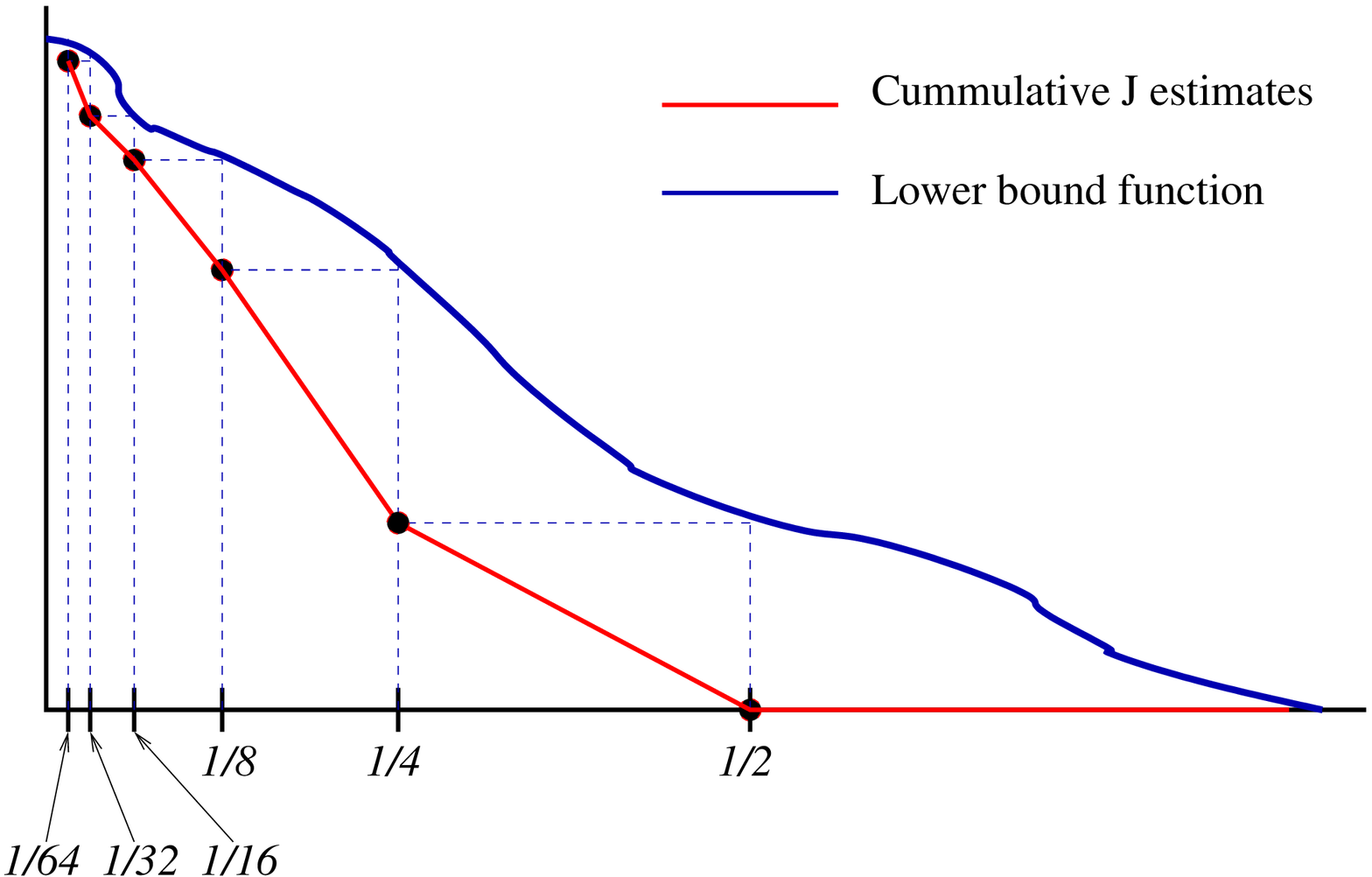,width=0.47\textwidth}
\fi
}
\caption{Lower bound function $\underline{f}^{(\vecv)}(u)$ for $u\in (0,1]$ and
 cumulative J estimates on outcomes consistent with $\vecv$
$\int_u^1 \hat{f}^{(J)}(x,\vecv)dx$.  The J estimate
$\hat{f}^{(J)}(u,\vecv)$ is the negated slope.
\label{Jest:fig}}
\end{figure}

\begin{lemma} \label{Jprop:lemma}
The J estimator is well defined,  is unbiased and nonnegative when
\eqref{nec_req} holds, and
satisfies
\begin{align}
\forall \rho  \forall \vecv,\ \int_\rho^1 \hat{f}^{(J)}(u,\vecv)du  &\le \underline{f}(\rho,\vecv) \label{leJ}\\
\forall \rho  \forall \vecv, \int_{\rho}^1 \hat{f}^{(J)} (u,\vecv) du &\ge \underline{f}(4\rho,\vecv)\ .\label{geJ}
\end{align}
\end{lemma}
\begin{proof}
We first argue that the constructions, which are presented relative to a
particular choices of the data $\vecv$, produce a consistent
estimator.  For that, we have to show that for every outcome $S(\rho,\vecv)$, the
assigned value is the same for all vectors $\vecz\in V^*(\rho,\vecv)$.
Since  $\underline{f}(\rho,\vecv) =
\underline{f}(\rho,\vecz)$ for all  $\vecz\in V^*(\rho,\vecv)$,
in particular this holds for $\rho=2^{-j}$, so
 the setting of the estimator for $u\in
(2^{-j-1},2^{-j}]$ is the same for all $S(u,\vecz)$ where
$\vecz\in V^*(2^{-j},\vecv)$ (also when $\vecz\not\in
V^*(2^{-j-1},\vecv)$). Therefore, the resulting estimator is
consistently defined.


We show that the construction maintains the following invariant for  $j\geq 1$:
\begin{equation} \label{inv1}
\underline{f}(2^{-j+1} ,\vecv) = \int_{2^{-j}}^1 \hat{f}^{(J)}(u,\vecv)du\ .
\end{equation}
From the first step of the construction,
$$\int_{1/2}^1 \hat{f}^{(J)}(u,\vecv) du = 
\int_{1/2}^1 2\underline{f}(1,\vecv) du = 
\underline{f}(1,\vecv)\ .$$
So \eqref{inv1} holds for $j=1$.
Now we assume by induction that \eqref{inv1} holds for $j$ and
establish that it holds for $j+1$.
If $\underline{f}(2^{-j},\vecv)= \underline{f}(2^{-j+1},\vecv)$,
then by the definition of the J estimator
$\int_{2^{-j-1}}^{2^{-j}} \hat{f}^{(J)}(u,\vecv)du=0$ and we get
$$\int_{2^{-j-1}}^1 \hat{f}^{(J)}(u,\vecv)du = \int_{2^{-j}}^1
\hat{f}^{(J)}(u,\vecv)du = \underline{f}(2^{-j+1},\vecv) = \underline{f}(2^{-j},\vecv)\ .$$
Otherwise, by definition,
$\int_{2^{-j-1}}^{2^{-j}} \hat{f}^{(J)}(u,\vecv)du=\underline{f}(2^{-j},\vecv)-\underline{f}(2^{-j+1},\vecv)$ and hence
$\int_{2^{-j-1}}^1 \hat{f}^{(J)}(u,\vecv)du
=\underline{f}(2^{-j},\vecv)$ and \eqref{inv1} holds for $j+1$.

From monotonicity, $\underline{f}(2^{-j+1},\vecv) \leq
\underline{f}(2^{-j},\vecv)$ and when substituting \eqref{inv1} in
the definition of the estimator we obtain that the estimates are always nonnegative.

To establish \eqref{leJ} we use \eqref{inv1}, the relation
$2^{\lfloor\log_2 \rho\rfloor} \leq \rho < 2^{1+\lfloor\log_2
  \rho\rfloor}$, and monotonicity of $\underline{f}(u,\vecv)$,
to obtain
\begin{align*}
& \int_\rho^1 \hat{f}^{(J)}(u,\vecv)du \leq  \int_{2^{\lfloor\log_2
    \rho\rfloor}} ^1 \hat{f}^{(J)}(u,\vecv)du = 
\underline{f}(2^{1 +
  \lfloor\log_2
    \rho\rfloor},\vecv) \leq \underline{f}(\rho,\vecv)\ .
\end{align*}

Similarly, we
establish \eqref{geJ} using
 \eqref{inv1}, and the relation
$2^{-1+\lceil \log_2   \rho\rceil}\leq  \rho \leq 2^{\lceil \log_2
  \rho\rceil}$:
\begin{align*}
\int_\rho^1 \hat{f}^{(J)}(u,\vecv) du & \ge \int_
{2^{\lceil \log_2\rho\rceil}}^1 \hat{f}^{(J)}(u,\vecv)du 
= 
  \underline{f}(2^{1+\lceil \log_2\rho\rceil},\vecv) \ge
  \underline{f}(4\rho,\vecv)
\end{align*}

Lastly, unbiasedness follows
 from \eqref{nec_req}
and combining \eqref{leJ} and \eqref{geJ}:
$$\underline{f}(\rho,\vecv)\ge \int_{\rho}^1 \hat{f}^{(J)}(u,\vecv) du \geq \underline{f}(4\rho,\vecv)\ .$$
  when we take the limit as $\rho\rightarrow 0$. \onlyinproc{\qed}
\end{proof}


\noindent
{\bf  Computing the J estimate from an outcome $S$:}  From the outcome
we know the seed value $\rho$ and 
the lower bound function $\underline{f}^{(\vecv)}(u)$ 
for all $u\geq \rho$
(recall that the lower bound on this range is the same for all data
$\vecv\in V^*(S)$, so we do not need to know the data $\vecv$). 
We compute $i\gets \lfloor -\log_2 \rho \rfloor$ and use the invariant \eqref{inv1} in
the definition of J, obtaining
the J estimate 
$2^{i+1}(\underline{f}(2^{-i},\vecv)-
\underline{f}(2^{-i+1},\vecv))$.

\smallskip
\noindent
{\bf  Example:}   We demonstrate the application of the J estimator
through a simple example.  The data domain in our example includes 
pairs $(v_1,v_2)$ of nonnegative
reals.  We are interested in $f(v_1,v_2)=(\max\{v_1-v_2,0\})^2$,
which sum aggregates to (the square of) the
 ``one sided''  Euclidean distance.  The data is PPS sampled
with threshold $\tau=1$ for both entries, therefore,
the sampling probability of entry $i$ is $\min\{1,
v_i\}$.  Sampling is coordinated, which means that for a seed $\rho\in
U[0,1]$,
entry $i$ is sampled if and only if $v_i\geq \rho$.  The outcome $S$ includes
the values of the sampled entries and the seed value $\rho$.
If no entry is sampled, or only the second entry is sampled, the
lower bound function for $x\geq \rho$ is $0$ and the
J estimate is $0$.    If only the first entry is sampled, the lower
bound function, for $x\geq \rho$, and accordingly, the J estimate are
\begin{align*}
\underline{f}^{(\vecv)}(x) &=
\max\{0, v_1-x\}^2 \\
\hat{f}^{(J)}(S) &= 
2^{\lfloor -\log_2 \rho \rfloor +1} \bigg(\max\{0, v_1- 2^{-\lfloor -\log_2
  \rho \rfloor}\}^2 \notinproc{\\
 &}  -\max\{0, v_1- 2^{1-\lfloor -\log_2
  \rho \rfloor}\}^2\bigg)\ .
\end{align*}
If both entries are sampled, the lower
bound function for $x\geq \rho$ and the $J$ estimate are
{\scriptsize
\begin{align*}
\underline{f}^{(\vecv)}(x) &=
\max\{0, v_1-\max\{v_2,x\}\}^2 \\
\hat{f}^{(J)}(S) &= 
2^{\lfloor -\log_2 \rho \rfloor +1}\bigg(\max\{0, v_1- \max\{v_2,2^{-\lfloor -\log_2
  \rho \rfloor}\}\}^2  
-\max\{0, v_1- \max\{v_2,2^{1-\lfloor -\log_2
  \rho \rfloor}\}\}^2\bigg)\ .
\end{align*}
}

\subsection{Competitiveness of the J estimator} \label{Jcomproof:sec}
\begin{theorem} \label{competitivehatf}
The estimator $\hat{f}^{(J)}$ is $O(1)$-competitive. 
\end{theorem}

\begin{proof}
We will show that
$$\forall \vecv,\, \int_0^1 \bigg(\hat{f}^{(J)}(u,\vecv)\bigg)^2du \leq 84 \int_0^1 \bigg(\hat{f}^{(\vecv)}(u)\bigg)^2du .$$

Let $\rho=2^{-j}$ for some integer $j\geq 0$. Recall the construction of $\hat{f}^{(J)}$ on an interval
$(\rho/2,\rho]$.   The value is fixed in the interval and is either
$0$, if $\int_\rho^1  \hat{f}^{(J)}(u,\vecv)du=\underline{f}^{(\vecv)}(\rho)$ or is
$2\frac{\underline{f}^{(\vecv)}(\rho)-\int_\rho^1
  \hat{f}^{(J)}(u,\vecv)du}{\rho}$.
Using \eqref{geJ} in Lemma~\ref{Jprop:lemma},
we obtain that for $u\in (\rho/2,\rho]$,
\begin{align*}
\hat{f}^{(J)}(u,\vecv) &\leq 2\frac{\underline{f}^{(\vecv)}(\rho)-\int_\rho^1  \hat{f}^{(J)}(u,\vecv)du}{\rho}\notinproc{\\
&}\leq 2\frac{\underline{f}^{(\vecv)}(\rho)-\underline{f}^{(\vecv)}(4\rho)}{\rho}
\end{align*}
Thus,
\begin{align}
\int_{\rho/2}^\rho \bigg( \hat{f}^{(J)}(u,\vecv)\bigg)^2 du &\leq \frac{\rho}{2}
\frac{4}{\rho^2}\bigg(\underline{f}^{(\vecv)}(\rho)-\underline{f}^{(\vecv)}(4\rho)\bigg)^2\notinproc{\nonumber\\
&}= \frac{2}{\rho}\bigg(\underline{f}^{(\vecv)}(\rho)-\underline{f}^{(\vecv)}(4\rho)\bigg)^2\label{hatfsquarebound}
\end{align}

We now bound the expectation of $\hat{f}^{(\vecv)}$ on $u\in
(\rho/2,4\rho)$  from below.
{\small
\begin{align}
\lefteqn{\int_{\rho/2}^{4\rho} \hat{f}^{(\vecv)}(u)du = 
\int_{\rho}^{4\rho} \hat{f}^{(\vecv)}(u)du +
\int_{\rho/2}^{\rho} \hat{f}^{(\vecv)}(u)du } \nonumber \\
\geq& \int_{\rho}^{4\rho} \hat{f}^{(\vecv)}(u)du +
\frac{\rho}{2}\hat{f}^{(\vecv)}(\rho) \label{jcompl3}\\
\geq& \int_{\rho}^{4\rho} \hat{f}^{(\vecv)}(u)du +
\frac{\rho}{2}\frac{\underline{f}^{(\vecv)}(\rho)-\int_\rho^1 \hat{f}^{(\vecv)}(u)du}{\rho} \label{jcompl4}\\
=& \int_{\rho}^{4\rho} \hat{f}^{(\vecv)}(u)du +
+ \frac{1}{2}\bigg(\underline{f}^{(\vecv)}(\rho)-\int_\rho^{4\rho} \hat{f}^{(\vecv)}(u)du-\int_{4\rho}^1 \hat{f}^{(\vecv)}(u)du\bigg) \nonumber\\
=& \frac{1}{2}\int_{\rho}^{4\rho}
\hat{f}^{(\vecv)}(u)du+\frac{1}{2}\bigg(\underline{f}^{(\vecv)}(\rho)-\int_{4\rho}^1
\hat{f}^{(\vecv)}(u)du\bigg) 
\geq \frac{1}{2}\bigg(\underline{f}^{(\vecv)}(\rho)-\underline{f}^{(\vecv)}(4\rho)\bigg) \label{jcompl9}
\end{align}
}
Inequality \eqref{jcompl3} follows from monotonicity of the function $\hat{f}^{(\vecv)}$.
Inequality \eqref{jcompl4}  from the definition of $\hat{f}^{(\vecv)}$ as the negated derivative of the lower hull of
  $\underline{f}^{(\vecv)}$  $\hat{f}^{(\vecv)}(\rho) \geq
  \frac{\underline{f}^{(\vecv)}(\rho)-\int_\rho^1
    \hat{f}^{(\vecv)}(u)du}{\rho}$ More precisely, we can use the
explicit definition of the $\vecv$-optimal estimates 
\eqref{condoptv}  provided in
  Section~\ref{vvaropt:sec}, 
$$\hat{f}^{(\vecv)}(\rho) = \inf_{0\leq \eta < \rho}
\frac{\underline{f}^{(\vecv)}(\eta)-\int_{\rho}^1 \hat{f}(u,\vecv)
  du}{\rho-\eta} \geq\inf_{0\leq \eta < \rho} \frac{\underline{f}^{(\vecv)}(\rho)-\int_{\rho}^1 \hat{f}(u,\vecv)
  du}{\rho-\eta}=\frac{\underline{f}^{(\vecv)}(\rho)-\int_{\rho}^1 \hat{f}(u,\vecv)
  du}{\rho}\ .$$  Lastly, inequality \eqref{jcompl9}
uses 
$\int_{4\rho}^1 \hat{f}^{(\vecv)}(u)du
\leq \underline{f}^{(\vecv)}(4\rho)$, which follows from
nonnegativity of $\hat{f}^{(\vecv)}$ and Lemma \ref{nonneg:lemma}.

Dividing both side by $3.5\rho$ we obtain a
lower bound
on  the {\em average} value of $\hat{f}^{(\vecv)}(u)$  in the
 interval $[\rho/2,4\rho]$.
\begin{align}\label{expfvintLB}
 &\frac{1}{3.5\rho} \int_{\rho/2}^{4\rho} \hat{f}^{(\vecv)}(u)du
\geq \frac{1}{7\rho}\bigg(\underline{f}^{(\vecv)}(\rho,\vecv)-\underline{f}^{(\vecv)}(4\rho)\bigg)
\end{align}

We next show that the value $\hat{f}^{(J)}(u,\vecv)$ on $u\in (\rho/2,\rho]$ is at most some
constant times the expected value of  the square of $\hat{f}^{(\vecv)}$ on $u\in
(\rho/2,4\rho)$.  
{\small
\begin{align}
& \int_{\rho/2}^{4\rho} \hat{f}^{(\vecv)}(u)^2du \geq 
\int_{\rho/2}^{4\rho} \left( \frac{1}{3.5\rho} 
\int_{\rho/2}^{4\rho}
\hat{f}^{(\vecv)}(u)du \right)^2 du \geq \label{ex2eqexs} \\
& \int_{\rho/2}^{4\rho} \left(
\frac{1}{7\rho}\bigg(\underline{f}^{(\vecv)}(\rho)-\underline{f}^{(\vecv)}(4\rho)\bigg) \right)^2 du \geq \label{pexpfvintLB}\\
& 3.5\rho\bigg(\frac{1}{7\rho}\bigg)^2\bigg(\underline{f}^{(\vecv)}(\rho)-\underline{f}^{(\vecv)}(4\rho)\bigg)^2 \geq\label{lastex2}\\
&  \frac{1}{28} \int_{\rho/2}^{\rho} \hat{f}^{(J)}(u,\vecv)^2du\nonumber
\end{align}}
Inequality \eqref{ex2eqexs}
uses the fact that for any random variable $X$,
$(\E[X])^2 \leq \E[X^2]$ applied to $\hat{f}^{(\vecv)}$ for $u\in
(\rho/2,4\rho]$.  Inequality \eqref{pexpfvintLB} follows from
\eqref{expfvintLB}.  Lastly, Inequality \eqref{lastex2} follows from \eqref{hatfsquarebound}.
We obtain
\begin{align*}
\int_0^1 \hat{f}^{(J)}(u)^2 du &= \sum_{i=0}^\infty \int_{2^{-i-1}}^{2^{-i}}
    \hat{f}^{(J)}(u)^2 du \\
&\leq 28 \sum_{i=0}^\infty \int_{2^{-i-1}}^{\min\{1,2^{-i+2}\}}
\hat{f}^{(\vecv)}(u)^2 du 
\leq 28\cdot 3 \int_0^1 \hat{f}^{(\vecv)}(u)^2  du
\end{align*}
\onlyinproc{\qed}
\end{proof}

\section{Proof of Theorem \ref{charact:thm}} \label{charactthmproof}

\vspace{0.05in}

\begin{proof}
{\bf $\bullet$  ``$\Rightarrow$''  \eqref{nec_req}:}
From Lemma~\ref{nonneg:lemma},
an unbiased and nonnegative estimator $\hat{f}$ must satisfy \eqref{nonneg:eq}.
Fixing $\vecv$ in \eqref{nonneg:eq} and taking the limit as
$\rho\rightarrow 0$ we obtain that 
$\E[\hat{f} | \vecv] = \int_0^1 \hat{f}(u,\vecv) du \leq
\lim_{u\rightarrow 0} \underline{f}^{(\vecv)} (u,\vecv)$. Combining with unbiasedness:
$\E[\hat{f} | \vecv]  =  f(\vecv)$ we obtain \eqref{nec_req}.

\smallskip
\noindent
{\bf $\bullet$ ``$\Leftarrow$'' \eqref{nec_req}}:
Follows immediately from Lemma~\ref{Jprop:lemma}.

\smallskip
\noindent
{\bf  $\bullet$ ``$\Rightarrow$''  \eqref{bounded_shared_nec_req}:} 
We bound from below
 the contribution to the expectation of unbiased and nonnegative 
$\hat{f}$ of outcomes $S(u,\vecv)$ for
 $u\leq \rho$:
$\int_0^\rho
\hat{f}(u,\vecv)=\int_0^1
\hat{f}(u,\vecv)- \int_\rho^1 \hat{f}(u,\vecv)\geq f(\vecv)-\underline{f}^{(\vecv)}(u)$.
The last inequality follows from unbiasedness and nonnegativity \eqref{nonneg:eq}.
Hence, the average value $\hat{f}(u,\vecv)$ when $u< \rho$ must be at least
$\frac{f(\vecv)-\underline{f}^{(\vecv)}(\rho)}{\rho}$,
and thus, considering all possible values of $\rho>0$,  we obtain that
$\hat{f}$ can be bounded only if  it satisfies \eqref{bounded_shared_nec_req}.

\smallskip
\noindent
{\bf $\bullet$  ``$\Leftarrow$''  \eqref{bounded_shared_nec_req}}: 
Note that \eqref{bounded_shared_nec_req} $\implies$
\eqref{nec_req}, and
therefore the conditions of Lemma~\ref{Jprop:lemma} are satisfied and
the J estimator is well defined, nonnegative,  and unbiased.  It
remains to show that given \eqref{bounded_shared_nec_req}, or the
equivalent statement
\begin{equation} \label{eqbounded_shared_nec_req}
\forall \vecv \, \exists c<\infty \, \forall u,\, f(v)-\underline{f}^{(\vecv)}(u)\leq
cu\ ,
\end{equation}
the J estimator is bounded.
Fix $\vecv$ and let $c$ be as in \eqref{eqbounded_shared_nec_req}.
 \begin{align}
 \hat{f}^{(J)}(\rho,\vecv) &\leq  2\frac{\underline{f}^{(\vecv)}(\rho/2)-
   \int_{2\rho}^1 \hat{f}^{(J)}(u,\vecv)du}{\rho}  \label{bsn1:ineq} \\ 
&\leq  2\frac{f(\vecv)- \underline{f}^{(\vecv)}(8\rho)}{\rho}  \label{bsn2:ineq} \\
&= 16 \frac{f(\vecv)- \underline{f}^{(\vecv)}(8\rho)}{8\rho} 
\leq
16 c \label{bsn3:ineq}
 \end{align}
Inequality \eqref{bsn1:ineq} is from the definition of the J
estimator.
Inequality \eqref{bsn2:ineq} uses definition of the lower bound
function and \eqref{geJ}.  Lastly, \eqref{bsn3:ineq} follows from our
assumption \eqref{eqbounded_shared_nec_req}.

\smallskip
\noindent
{\bf $\bullet$ ``$\iff$'' \eqref{bounded_var_shared_nec_req}}:
From Theorem~\ref{voptlh}, for all $\vecv$, 
\eqref{bounded_var_shared_nec_req}, which is square-integrability of $\hat{f}^{(\vecv)}(u)$,
 is {\em necessary} for existence of a nonnegative unbiased  estimator with
finite variance for $\vecv$.  Sufficiency follows from the proof of
Theorem~\ref{competitivehatf}, which shows that for all $\vecv$, the
expectation of the square of
the J estimator is at most a constant times the minimum possible,  and
\eqref{ccompvar}, which states that the variance is bounded if and only
if the expectation of the square is bounded. \onlyinproc{\qed}
\end{proof}

\subsection*{Conclusion}
We developed a precise understanding 
of the queries we can estimate accurately over coordinated shared-seed
samples.  We
defined variance competitiveness, and showed that $O(1)$
competitiveness is always attainable when estimators with finite
variances exist.
Our work uses a fresh, CS-inspired, and unified approach to the study of  estimators that is
particularly suitable for data analysis from samples and sets the
ground for continued work.

In a follow up work \cite{CKsorder:2012}, aiming for
good performance in practice, we
seek estimators that are  {\em variance optimal} ({\em admissible}
  \cite{surveysampling2:book}), that is, can not
be strictly improved.   
We show that there is
a range of variance optimal and competitive estimators and that this
choice can be leveraged to obtain 
natural estimators with desired properties or customized estimators to
patterns in our data:
 The L$^*$ estimator, which is defined for any function for which a
 nonnegative unbiased estimator exists,   is the
unique monotone\footnote{Estimate value is non-increasing with
  information}  variance-optimal estimator and
is guaranteed to have a ratio of at most $4$.   Customization
is facilitated by deriving {\em order-optimal
estimators} which optimize variances with respect to (any) 
specified  priorities over the data domain. The L$^*$ estimator
prioritizes data on which the estimated function is smaller 
whereas the U$^*$ estimator prioritizes large values.
We demonstrate the potential for applications
in~\cite{CK_sdiff_arxiv:2012},
for $L_p$ difference estimation and in
\cite{CDFGGW:COSN2013}, for 
 sketch-based similarity estimation in massive graphs.

  A natural question is to bound the best possible competitive ratio.
That is, the supremum over instances (data domain, shared-seed
sampling scheme, and function) for which unbiased nonnegative estimators with finite
variances exist of the best ratio attainable for the instance.
The L$^*$ estimator in \cite{CKsorder:2012} gives an upper bound of
$4$ and we were able to provide a tighter upper bound of $3.375$
and a lower bound of $1.44$.  An interesting question is to close this gap.

In the future, we hope  to build on our understanding of coordinated
sampling to extend our treatment of independent sampling in \cite{CK:pods11}. 
We envision automated tools that provide estimates according
to specifications of the sampling scheme, query, competitive ratio, 
and prioritization of patterns in the data.

\small
\bibliographystyle{plain}
\bibliography{cycle,replace,p2p,data_structures,varopt}

\newpage
\appendix

\section{Pointwise variance optimality} \label{vvaropt:sec} 

We work with {\em partial specification} of (nonnegative unbiased)
estimators.  A partial specification $\hat{f}$ is defined over
a subset ${\cal S}$ of outcomes that is 
closed under the containment order $V^*(S)$:
\begin{align*}
\forall \vecv\ \exists\rho_v\in [0,1],\ &
S(u,\vecv)\in {\cal S}\, \mbox{{\it a.e.} for }  u>\rho_v\, \wedge\,
\\
& S(u,\vecv)\not\in {\cal S}\, \mbox{{\it a.e.} for }  u\leq \rho_v \ .
\end{align*}
We also require  that  $\hat{f}:{\cal S}\geq 0$, that is, estimates
are nonnegative when specified,  and
that 
\begin{subequations}\label{pdnereq}
\begin{eqnarray}
\forall \vecv,\ \rho_v>0 & \implies & \int_{\rho_v}^1 \hat{f}(u,\vecv)du \leq f(\vecv) \label{condnonneg}\\
\forall \vecv,\ \rho_v=0 & \implies & \int_{\rho_v}^1 \hat{f}(u,\vecv)du = f(\vecv) \label{condunbiased}
\end{eqnarray}
\end{subequations}
If $\rho_v=0$, we say that the estimator is
{\em fully specified} for $\vecv$.

 We show that  a partial specification $\hat{f}$  can always be 
extended to an unbiased nonnegative estimator:
\begin{lemma} \label{completioncoro}
If $f$ satisfies \eqref{nec_req} (has a nonnegative unbiased
estimator) and $\hat{f}$ is partially specified then
it can be extended to an unbiased nonnegative estimator.
\end{lemma}
\begin{proof}
The only if direction follows from Lemma~\ref{nonneg:lemma}.  The ``if'' direction
follows from
a slight modification of the J estimator construction, in which
we ``start'' to specify the estimator at $\rho_v$
instead of $1$. \onlyinproc{\qed}
\end{proof}


\begin{theorem} \label{precoptcompletion}
If $\hat{f}$ is a partially specified estimator, then
an extension of $\hat{f}$
to the outcomes $S(\rho,\vecv)$ for $\rho\in (0,\rho_v]$ is a
partially specified estimator and
minimizes $\var[\hat{f}|\vecv]$ (over all unbiased and nonnegative
estimators which are extensions of the partial specification $\hat{f}$).
\begin{eqnarray}
\lefteqn{\iff\ \text{for $u \in (0,\rho_v]$ {\it a.e.}}}\nonumber \\
\hat{f}(u,\vecv) &=& \inf_{0\leq \eta < u}
\frac{\underline{f}^{(\vecv)}(\eta)-\int_{u}^1 \hat{f}(x,\vecv)
  dx}{u-\eta}\ . \label{condoptv}
\end{eqnarray}
\ignore{
\begin{align} \label{condoptv2}
\iff \forall \rho\in (0,\rho_v],\ \quad\quad\quad \nonumber\\
\lim_{\eta\rightarrow \rho^{-}} \frac{\int_\eta^\rho \hat{f}(x,\vecv)dx}{\rho-\eta} = \inf_{0\leq \eta < \rho}
\frac{\underline{f}^{(\vecv)}(\eta)-\int_{\rho}^1 \hat{f}(x,\vecv)
  dx}{\rho-\eta}\ .
\end{align}
}
Moreover, a solution $\hat{f}$ of \eqref{condoptv} always exists and
$H(u)=\int_u^1 \hat{f}(x,\vecv) dx$ for $u\in (0,\rho_v]$ is unique.
\end{theorem}
\begin{proof}
\ignore{The equivalence
\eqref{condoptv2} $\iff$ \eqref{condoptv} follows from
Lemma \ref{lebesguediff}, so it suffices to
establish the claim for \eqref{condoptv}.}


From Lemma \ref{completioncoro}  an extended $\hat{f}$ to
 outcomes $S(\eta,\vecv)$ for $\eta\in (0,\rho_v]$ is
a partially specified estimator if and only if
\begin{subequations} \label{nec_opt}
\begin{eqnarray}
\forall u\leq \rho_v\ 
\int_u^1 \hat{f}(x,\vecv)dx \leq \underline{f}^{(\vecv)}(u)\  \label{neg_prec}\\
\int_0^1 \hat{f}(x,\vecv)dx=f(\vecv) \label{unb_prop}
\end{eqnarray}
From the structure of the constraints \eqref{nec_opt},  if we take a
solution $\hat{f}(u,\vecv)$  and shift ``weight'' to lower
$u$ values (decreasing $\hat{f}$ on higher $u$ values and increasing on
lower ones so that the total contribution is the same) then the
result is also a solution.  If the shift results in decrease of higher
values and increase of lower values, then the variance decreases.  Therefore, a
minimum variance solution of \eqref{nec_opt} must be monotone
nonincreasing with $u$ almost everywhere, meaning, that monotonicity
holds if we exclude
from $(0,\rho_v]$ a set of points with zero measure.
\begin{align}
\exists \text{zero measure $U^0\subset (0,\rho_v]$}  \forall u_1,u_2\in
(0,\rho_v]\setminus U_0, \nonumber \\
  u_1<u_2 \implies \hat{f}(u_1,\vecv) \geq \hat{f}(u_2,\vecv) \label{monotae}
\end{align}

Looking now only at extensions satisfying
\eqref{neg_prec}, \eqref{unb_prop}, and \eqref{monotae}, we can decrease
variance if we could "shift weight''
from the higher estimates on
lower $u$ values to the lower estimates on higher $u$ values.   Such shifts,
however, may violate \eqref{neg_prec}.  Therefore, an estimator has minimum
variance only if all such shifts result in violation,
that is, after excluding a zero measure set, between
any two points 
with strictly different estimate values there must be
a point $\eta$ such that $\int_\eta^1 \hat{f}(u,\vecv)du$
``touches'' the lower bound at $\eta$.
Formally,
\begin{eqnarray}
\exists \text{ zero measure $U^0\subset (0,\rho_v]$}  \forall
u_1,u_2\in (0,\rho_v]\setminus U_0 \nonumber \\
  u_1< u_2  \wedge \hat{f}(u_1,\vecv)>
\hat{f}(u_2,\vecv) & \implies &\nonumber\\
 \exists \eta\in (u_1,u_2],
\int_\eta^1 \hat{f}(u,\vecv)du = \lim_{u\rightarrow \eta^+} \underline{f}^{(\vecv)}(u)\ .&& \label{touchLB}
\end{eqnarray}
\end{subequations}


\ignore{
An equivalent requirement to \eqref{neg_prec} is that
for all $\eta<\rho\leq \rho_v$, the expectation of $\hat{f}(u,\vecv)$ for $u\in
(\eta,\rho]$ does not exceed the ratio
$\frac{\underline{f}^{(\vecv)}(\eta) - \int_\rho^1 \hat{f}(u,\vecv)du}{\rho-\eta}$.
since it satisfies the infimum of these requirements.
}

  We established that \eqref{nec_opt} (\eqref{neg_prec}- \eqref{touchLB})
are {\em necessary} for an extension to have minimum  variance.
We also argued that any extension satisfying \eqref{neg_prec} and
\eqref{unb_prop} can be transformed to one satisfying \eqref{monotae} and
\eqref{touchLB} without increasing its variance.  Therefore, if we can
establish that \eqref{nec_opt} have a unique solution (up to
equivalence), the solution must have minimum variance.


  Geometrically, an extension $\hat{f}$ satisfies either
  \eqref{nec_opt}  or \eqref{condoptv} if and only if 
the function $H(u)=\int_u^1 \hat{f}(x,\vecv) dx$ is 
 the lower boundary of the convex hull of the lower bound function
$\underline{f}^{(\vecv)}(x)$ for $x\in (0,\rho_v]$ and the point
$(\rho_v,\int_{\rho_v}^1 \hat{f}(u,\vecv) du)$.
This implies existence and uniqueness of the solution and also
that it has the desired properties.

  We provide an algebraic proof.  We start by
exploring the structure which an extension of $\hat{f}(u,\vecv)$
  to $u\in (0,\rho_v]$, which satisfies
 \eqref{condoptv}, must have.  We subsequently use the structure
 to establish existence and nonnegativity,  uniqueness ({\em a.e.}),  that it satisfies
\eqref{nec_opt}, and lastly, to show that any
 function that violates \eqref{condoptv}  must also violate
 \eqref{nec_opt}.  The combination of these properties implies the
 claim of the theorem.  


\noindent
 {\bf Structure:}
We say that $\rho$ is an {\em LB (lower bound) point} of
$\hat{f}(u,\vecv)$
if the infimum of the
solution of \eqref{condoptv} at $\rho$ is achieved as $\eta\rightarrow
\rho ^-$, that is,
$$\hat{f}(\rho,\vecv) = \lim_{\eta\rightarrow
  \rho^-}\frac{\underline{f}^{(\vecv)}(\eta)-\int_\rho^1 \hat{f}(u,\vecv)du}{\rho-\eta}\ .$$
 Otherwise, if the infimum is
approached when  $\eta < \rho-\epsilon$ for some $\epsilon>0$, we
say that $\rho$ is a
 {\em gap point}.
Note that it is possible for $\rho$ to be both an LB and a gap point
if the infimum is approached at multiple places.
If $$\int_\rho^1 \hat{f}(u,\vecv)du < \lim_{\eta\rightarrow \rho^-}
\underline{f}^{(\vecv)}(\eta)=\underline{f}^{(\vecv)}(\rho)\ ,$$ then $\rho$
must be a gap point, and if equality holds, $\rho$ can be either an LB
or a gap point (or both).
The equality follows using Corollary~\ref{leftcon:lemma} (left continuity of
$\underline{f}^{(\vecv)}$).


We use this classification of points to partition
$(0,\rho_v]$ into {\em LB} and {\em gap} subintervals of the form
$(y,\rho]$, that is, open to the left and
closed to the right.

{\em LB subintervals} are maximal subintervals
containing exclusively LB points (which can
double as gap points) that have the form $(y,\rho]$.  
From left continuity and monotonicity of $\underline{f}^{(\vecv)}$,
if $\rho$ is an LB point and not a gap point then
there is some $\epsilon>0$ such that $(x-\epsilon,x]$
are also LB points.  Thus all LB points that are not gap points must
be part of LB subintervals and these subintervals  are open to the left and
closed to the right (the point $y$ is a gap point which may double as an
LB point).  Alternatively, we can identify
LB subintervals as maximal
such that
$$\forall x\in (y,\rho],\ \int_x^1 \hat{f}(u,\vecv)du =
\underline{f}^{(\vecv)}(x)\ .$$

 Gap subintervals are maximal that satisfy
\begin{eqnarray}
&& \forall x\in (y,\rho),\ \int_x^1 \hat{f}(u,\vecv)du < \lim_{u\rightarrow
 x^{+}} \underline{f}^{(\vecv)}(u) \label{midprop} \ .
\end{eqnarray}
Note that a consecutive interval of gap points may
consist of multiple back-to-back gap subintervals.

We can verify that the  boundary points $b\in \{y, \rho\}$ of both LB and gap subintervals
satisfy
\begin{equation}\label{boundarycond}
 \int_b^1 \hat{f}(u,\vecv)du = \lim_{x\rightarrow
 b^+} \underline{f}^{(\vecv)}(x) 
\end{equation}
Visually, LB intervals are segments where $\int_u^1 \hat{f}(x,\vecv)dx$ ``identifies'' with
 the lower bound function whereas gap intervals are linear
segments where it ``skips'' between two points where it touches (in
the sense of limit from the right) the (left-continuous) lower bound function.
\ignore{
 We can verify that the  boundary points $y$ and $\rho$ of LB and gap subintervals
satisfy
\begin{subequations}\label{boundarycond}
\begin{eqnarray}
&& \int_\rho^1 \hat{f}(u,\vecv)du = \lim_{x\rightarrow
 \rho^+} \underline{f}^{(\vecv)}(x) \label{rhoprop} \\
&& y = \arg\min_{x\in [y,\rho)} \int_x^1 \hat{f}(u,\vecv)du = \lim_{u\rightarrow
 x^{+}} \underline{f}^{(\vecv)}(u)\ . \label{yprop}
\end{eqnarray}
\end{subequations}
}
Figure~\ref{LBopt:fig} illustrates a partition of an example lower bound
function and its lower hull into gap and LB subintervals.
From left
to right, there are two gap subintervals, an LB subinterval, a gap subinterval, and finally a trivial LB subinterval (where the LB and estimates are $0$).

\noindent
{\bf Existence, nonnegativity, and \eqref{neg_prec} :}

  For some $0< \rho\leq \rho_v$, let 
$\hat{f}(u,\vecv)$ be an extension which
satisfies \eqref{condoptv} 
and  \eqref{neg_prec}
for all
$u\in (\rho,\rho_v]$.
We show that there exists $y< \rho$ such that the solution can be extended to 
$(y,\rho]$, that is, \eqref{condoptv} 
and  \eqref{neg_prec} are satisfied also for $u\in (y,\rho]$.

Consider the solution
of \eqref{condoptv} at $u=\rho$.
If the infimum is attained at $\eta<\rho$,
let $y$ be the infimum over points $\eta$ that attain
the infimum of
\eqref{condoptv} at $\rho$.  We can extend the
solution of  \eqref{condoptv} to the interval $(y,\rho]$, which is a
gap interval (or the prefix of one).  The solution is fixed throughout
the interval:
\begin{equation}\label{fixedgap}
\forall x\in (y,\rho],\ \hat{f}(x,\vecv)=\frac{\displaystyle{\lim_{\eta\rightarrow
 y^{+}}} \underline{f}^{(\vecv)}(\eta)-\displaystyle{\lim_{\eta\rightarrow \rho^+}}\underline{f}^{(\vecv)}(\eta)}{\rho-y}\ .\end{equation}

 Assume now that
 the infimum of the solution
of \eqref{condoptv} at $\rho$ is attained
as $\eta\rightarrow \rho^{-}$ and not at any
 $\eta<\rho$.
Let  $y$ be the supremum of points $x<\rho$
{\small
\begin{equation}\label{LBchoice}
\inf_{0\leq \eta < x}
\frac{\underline{f}^{(\vecv)}(\eta)-\lim_{u\rightarrow
    x^+}\underline{f}^{(\vecv)}(u)}{x-\eta} <
\lim_{\eta \rightarrow x^-}
\frac{\underline{f}^{(\vecv)}(\eta)-\lim_{u\rightarrow
    x^+}\underline{f}^{(\vecv)}(u)}{x-\eta}
\end{equation}
}
(We allow the rhs to be $+\infty$).
We have $y<\rho$ and can extend the solution to $(y,\rho]$.  The
extension satisfies $\int_x^1 \hat{f}(u,\vecv)du =
\underline{f}^{(\vecv)}(x)$ and $(y,\rho]$  is an LB subinterval or a prefix of one.
The actual estimator values on $(y,\rho]$ are the negated left
derivative of $\underline{f}^{(\vecv)}(u)$:
$$\forall x\in (y,\rho],  \hat{f}(x,\vecv)=\lim_{\eta \rightarrow x^-}
\frac{\underline{f}^{(\vecv)}(\eta)-\underline{f}^{(\vecv)}(x)}{x-\eta}$$


 In both cases (where the extension corresponds to an LB or gap subinterval),
the solution on $(y,\rho]$ exists,
is nonnegative,
and satisfies \eqref{neg_prec} for $u\in (y,\rho]$.

  Thus,
starting from $u=\rho_v$, we can iterate the process
and compute a solution on a suffix of $(0,\rho_v]$
that is partitioned into gap and LB intervals.  The sum of sizes of
intervals, however, may converge to a value $<\rho_v$
and thus the process may not converge to covering $(0,\rho_v]$.
To establish existence, assume to the contrary that there is no
solution that covers $(0,\rho_v]$. Let
$x$ be the infimum such that there are solutions on $(x,\rho_v]$ but
there are no solutions for $(x-\epsilon,\rho_v]$ for all $\epsilon>0$.
Let $\hat{f}(u,\vecv)$ be a solution on $u\in (x,\rho_v]$ and consider its
partition to LB and gap intervals.
Each interval that intersects $[x,x+\epsilon)$ must contain a left boundary point $y_\epsilon$ of
some subinterval (it is possible that $y_\epsilon=x$).
All boundary points satisfy
$\int_{y_\epsilon}^1 \hat{f}(x,\vecv)\leq
\underline{f}^{(\vecv)}(y_{\epsilon})$.
We have
\begin{eqnarray*}
\int_{x}^1 \hat{f}(u,\vecv)du & = &  \lim_{\epsilon\rightarrow
  0} \int_{y_\epsilon}^1 \hat{f}(x,\vecv)\leq \lim_{\epsilon\rightarrow 0} \underline{f}^{(\vecv)}(y_{\epsilon})\\
&=&  \lim_{\eta\rightarrow x^{+}} \underline{f}^{(\vecv)}(\eta) \leq
\underline{f}^{(\vecv)}(x)
\end{eqnarray*}
The last inequality follows from monotonicity of
$\underline{f}^{(\vecv)}(u)$.  Thus, we can apply our extension process
starting from $\rho=x$ and
extend the solution to $(y,x]$ for some $y<x$,
contradicting our assumption and establishing existence.


\noindent
{\bf Uniqueness:}
Assume there are two different solutions and let $x\in (0,\rho_v)$
be the supremum of values on which they differ.  Consider one of the
solutions.
$x$ can not be an interior point of a (gap or LB) subinterval because
the estimator is determined on the subinterval from values on the
right boundary, which is the same for both solutions.  If $x$ is a
right boundary point, $\int_x^{\rho_v} \hat{f}(u,\vecv)du$ is same
for both functions and thus the left boundary and the solution on the
interval are uniquely determined and must be the same for both
functions, contradicting our assumption.


\ignore{
note that it is satisfied at $\rho=\rho_v$ and at points $x$ where
there is equality, $\hat{f}^{(\vecv)}(x,\vecv)$ does not exceed the left partial
derivative of $\underline{f}(x,\vecv)$ (which exists from
left-continuity and monotonicity of $\underline{f}$):
\begin{eqnarray*}
\hat{f}^{(\vecv)}(\rho,\vecv) &= &
\inf_{0\leq \eta<\rho}\frac{\underline{f}(\eta,\vecv')-\underline{f}(\rho,\vecv)}{\rho-\eta}  \\
& \leq & \lim_{\eta\rightarrow \rho^{-}} \frac{\underline{f}(\eta,\vecv')-\underline{f}(\rho,\vecv)}{\rho-\eta} =
-\frac{\partial \underline{f}(x,\vecv)}{\partial x^{-}}
\end{eqnarray*}
}

\noindent
{\bf Satisifes \eqref{unb_prop} (unbiasedness):}
From \eqref{nec_req} (a sufficient necessary condition to
existence of unbiased nonnegative estimator) and \eqref{boundarycond}
$$\int_0^1 \hat{f}(u,\vecv)du = \lim_{u\rightarrow 0^+}
\underline{f}^{(\vecv)}(u)=f(\vecv).$$

\noindent
{\bf Satisifes \eqref{monotae} and  \eqref{touchLB}:}
  Within an LB subinterval the estimator must be
  nonincreasing,
since otherwise we obtain a contradiction to the definition of an LB
point.  The fixed values on two consecutive gap intervals  $(y,x]$ and
$(x,\rho]$ must also
be nonincreasing because otherwise the infimum of the solution at
$\rho$ could not be obtained at $x$ since a strictly lower value is
obtained at $y$.

\noindent
{\bf The  {\em only} solution of \eqref{nec_opt}:}  Let
$\hat{f}^{(\vecv)}(u)\equiv \hat{f}(u,\vecv)$ where $\hat{f}$ is a
solution of \eqref{condoptv}.  Let 
$\hat{f}'$  be another extension  which satisfies
\eqref{neg_prec}  and \eqref{unb_prop} (and thus constitutes
a 
partially specified estimator that is fully specified on $\vecv$)
 but  is not equivalent to $\hat{f}^{(\vecv)}$, that is, for some
$\rho$
$$\int_\rho^{\rho_v} \hat{f}^{(\vecv)}(u)du \not= \int_\rho^{\rho_v}
\hat{f}'(u,\vecv)du \ .$$
We show that $\hat{f}'$ must violate
\eqref{monotae} or  \eqref{touchLB}.

We first show that if
\begin{equation} \label{largecase}
\int_\rho^{\rho_v} \hat{f}^{(\vecv)}(u)du < \int_\rho^{\rho_v}
\hat{f}'(u,\vecv)du
\end{equation}
then monotonicity \eqref{monotae} must be violated.
Let $\overline{\rho}$ be the supremum of points satisfying \eqref{largecase}.
Then $\overline{\rho}$ must be a gap point
and some neighborhood to its left must be gap points.
Let $y< \overline{\rho}$ be the left boundary point of this gap interval.
Recall that
$\hat{f}^{(\vecv)}$ is constant on a gap interval.

 Let $\rho\in(y,\overline{\rho})$ be such that \eqref{largecase} holds.
There must be measurable set in $[\rho,\overline{\rho}]$ on which
$\hat{f}' > \hat{f}^{(\vecv)}$.
Since $y$ is a left boundary of a gap interval, from
\eqref{boundarycond}, 
$$\int_y^1 \hat{f}^{(\vecv)}(u)du=\lim_{u\rightarrow y^+}\underline{f}^{(\vecv)}(u)$$ and
therefore, since $\hat{f}'$ satisfies \eqref{neg_prec},
$\int_y^{\rho_v} \hat{f}'du \leq \int_y^{\rho_v} \hat{f}^{(\vecv)}(u)du$.
Thus, there must exist a measurable subset of
$[y,\rho]$ on which
$\hat{f}'< \hat{f}^{(\vecv)}$.  Since $\hat{f}^{(\vecv)}$ is
constant in $(y,\overline{\rho}]$, $\hat{f}'$ violates \eqref{monotae}.

Lastly, we show that
if
\begin{equation} \label{smallcase}
\int_\rho^{\rho_v} \hat{f}^{(\vecv)}(u)du > \int_\rho^{\rho_v}
\hat{f}'(u,\vecv)du
\end{equation}
then
even
if \eqref{monotae} holds then \eqref{touchLB} is violated.

 There must exist a measurable subset $U\subset [\rho,\rho_v]$ on which
$\hat{f}^{(\vecv)}(u)> \hat{f}'(u,\vecv)$.

 Let $\eta$ be the supremum of points $x\in [0,\rho)$ satisfying
$\int_x^{\rho_v} \hat{f}'(u,\vecv)du=\lim_{u\rightarrow x^+}
\underline{f}^{(\vecv)}(u)$.  Since this is satisfied by $x=0$ and
the supremum can not be satisfied by $\rho$, such $\eta< \rho$ is well defined.

 Since $\hat{f}^{(\vecv)}$ satisfies \eqref{neg_prec},
$\int_\eta^{\rho_v} \hat{f}^{(\vecv)}(u)du \leq
\int_\eta^{\rho_v} \hat{f}'(u,\vecv)du$.
Therefore,
there must be another measurable set $W\subset
[\eta,\rho]$ on which
$\hat{f}^{(\vecv)}(u)< \hat{f}'(u,\vecv)$.
Because $\hat{f}^{(\vecv)}$ is monotone, the values of $\hat{f}'$ on
$W$ are strictly larger than on $U$.  Hence, $\hat{f}'$ violates
\eqref{touchLB}.
\end{proof}


Theorem~\ref{voptlh} follows by applying Theorem~\ref{precoptcompletion}  to
an empty specification.


\section{Using sum estimators with bottom-$k$
  sampling}  \label{sumcase:sec}
We provide details on the application of sum estimators with bottom-$k$ sampling of instances.
A technical requirement for applying sum estimators is that
the sampling scheme of each item is explicit.
This naturally happens with Poisson sampling
where an entry (item $h$ in instance $i$) is sampled $\iff$
 $r(u(h),v(h)) \geq T_i(h)$.
 With bottom-$k$ sampling, however, item inclusions  are dependent and
 therefore the separation a bit more technical, and relies on a technique
named rank
conditioning~\cite{bottomk:VLDB2008,dlt:pods05}.

Rank conditioning was
applied to estimate subset sums over bottom-$k$ samples of a
single instance, and 
allows us to treat each entry as Poisson sampled.
In rough details, for each entry, we obtain a  ``$T_i(h)$  substitute''  which allows
the inclusion of $h$ in the sample of instance $i$ 
on seeds (and thus ranks) of all other items being
fixed.  With this conditioning,
the inclusion of $h$ in the sample of $i$ is Poisson with threshold
$T_i$, that is, follows the rule $r(u(h),v(h)) \geq T_i$, where $T_i$
is  equal to the $k$th largest rank value of items in 
instance $i$ with $h$
excluded which is the same as the $(k+1)$st largest rank value.

\end{document}